%% file: vectormonotonicity.tex
\documentclass[11pt]{article}
\usepackage{listings}
\usepackage{cancel}
\usepackage{bbm}
\usepackage{graphicx}
\usepackage{subcaption}
\usepackage{mathtools}
\usepackage{amsmath}
\usepackage{amsfonts}
\usepackage{enumerate}
\usepackage{xr}
\usepackage{varioref}
\usepackage{xr-hyper}
\usepackage{hyperref}
\usepackage{multirow}

\usepackage{amsthm} 
\usepackage[toc,page]{appendix}
\usepackage[font=small,labelfont=bf, skip=0pt]{caption}

\usepackage{array}
\newcolumntype{H}{>{\setbox0=\hbox\bgroup}c<{\egroup}@{}}

\usepackage[a4paper,margin=1in,footskip=0.25in]{geometry}

\usepackage[utf8]{inputenc}
\usepackage[style=authoryear-comp,firstinits=true,citestyle=authoryear,natbib=true,backend=bibtex,doi=false,isbn=false,url=false,maxcitenames=2,eprint=false]{biblatex}
\AtBeginBibliography{}

\DeclareNameAlias{sortname}{last-first}

\renewbibmacro{in:}{}
\renewbibmacro*{volume+number+eid}{%
	\printfield{volume}%
	\setunit*{\addnbspace}
	\printfield{number}%
	\setunit{\addcomma\space}%
	\printfield{eid}}
\DeclareFieldFormat[article]{number}{\mkbibparens{#1}}

\bibliography{biblio}

\usepackage{nameref, zref-xr,zref-hyperref,zref-user}

\makeatletter
\def\thmhead@plain#1#2#3{%
	\thm@notefont{}
	\thmname{#1}\thmnumber{\@ifnotempty{#1}{ }\@upn{#2}}%
	\thmnote{ {\the\thm@notefont#3}}}
\let\thmhead\thmhead@plain
\itshape 
\makeatother

\newtheorem*{definition*}{Definition}
\newtheorem{definition}{Definition}
\newtheorem*{assumption*}{Assumption}
\newtheorem*{lemma*}{Lemma}
\newtheorem{lemma}{Lemma}
\newtheorem*{sublemma*}{Sublemma}

\newtheorem*{proposition*}{Proposition}
\newtheorem{proposition}{Proposition}
\newtheorem*{conjecture*}{Conjecture}
\newtheorem*{theorem*}{Theorem}
\newtheorem{theorem}{Theorem}
\newtheorem*{corollary*}{Corollary}

\newcommand{\indep}{\perp\!\!\!\!\perp}

\usepackage{lipsum}
\usepackage{titlesec}
\titleformat*{\subsubsection}{\large\bfseries}

\graphicspath{ {images/} }

\usepackage[many]{tcolorbox}
\newsavebox{\fmbox}

\usepackage[doublespacing]{setspace}

\title{A Vector Monotonicity Assumption for Multiple Instruments}
\author{Leonard Goff\thanks{\protect\linespread{1}\protect\selectfont Department of Economics, University of Calgary. 2500 University Dr. N.W. Calgary, AB T2N1N4, Canada. Email: \texttt{leonard.goff@ucalgary.ca}. I am grateful to Josh Angrist, Simon Lee, Suresh Naidu and Bernard Salani\'e for insightful feedback throughout this project, which was a chapter of my PhD dissertation. I also thank Isaiah Andrews, Jushan Bai, Junlong Feng, Peter Hull, Jack Light, Jos\'e Luis Montiel Olea, Serena Ng, Vitor Possebom and Alex Torgovitsky for helpful comments and discussion. I thank attendees of the Columbia econometrics colloquium, the 2019 Young Economists Symposium, the 2019 Empirics and Methods in Economics Conference, the 2022 Georgia Econometrics Workshop, and the 2022 North American Meeting of the Econometric Society for their feedback. I thank the Bureau of Labor Statistics for data access. Any errors are my own. Online Appendices are available \href{http://www.leonardgoff.com/resources/vm_externalappendix.pdf}{\nolinkurl{here}}.}}

\date{}

\begin{document}

\maketitle

\begin{abstract}
When a researcher combines multiple instrumental variables for a single binary treatment, the monotonicity assumption of the local average treatment effects (LATE) framework can become restrictive: it requires that all units share a common direction of response even when separate instruments are shifted in opposing directions. What I call \textit{vector monotonicity}, by contrast, simply assumes treatment uptake to be monotonic in all instruments. I characterize the class of causal parameters that are point identified under vector monotonicity, when the instruments are binary. This class includes, for example, the average treatment effect among units that are in any way responsive to the collection of instruments, or those that are responsive to a given subset of them. The identification results are constructive and yield a simple estimator for the identified treatment effect parameters. An empirical application revisits the labor market returns to college.
\end{abstract}

\newpage

\normalsize

\input{vm_body}


\printbibliography

\begin{appendices}
	\input{vm_appendices}
	\input{vm_mtwappendix}
	\input{vm_proofs}
\end{appendices}

\end{document}

%% file: vm_body.tex
\section{Introduction}
The local average treatment effects (LATE) framework of \citet{Imbens2018} allows instrumental variables to be used for causal inference even when there is arbitrary heterogeneity in treatment effects. However, the model makes an important assumption about homogeneity in individuals' selection behavior, referred to as \textit{monotonicity}. When the researcher has a single instrumental variable at their disposal, this LATE monotonicity assumption is typically quite a natural one to make. But when multiple instruments are combined, LATE monotonicity can become hard to justify---a point recently emphasized by \citet*{Mogstad2020a} (henceforth MTW). 

This paper considers a natural alternative assumption, which is that monotonicity holds on an instrument-by-instrument basis: what I call \textit{vector monotonicity} (VM).\footnote{A leading special case of VM is discussed by MTW under the name \textit{actual monotonicity} (see Section \ref{sec:notions}).} Vector monotonicity assumes that each instrument has an impact on treatment uptake in a direction that is common across units, regardless of the values of the other instruments. This direction need not be known ex-ante by the researcher, but is often implied by economic theory. For example, two instruments for college enrollment might be: i) proximity to a college; and ii) affordability of nearby colleges. VM assumes that increasing either instrument induces some individuals towards going to college, while discouraging none, i.e. proximity to a college weakly encourages college attendance regardless of price, and lower tuition weakly encourages college attendance regardless of distance. This contrasts with the traditional monotonicity assumption of the LATE model, which requires that either proximity or affordability dominates in the selection behavior of all individuals: in particular, it implies that all individuals who would go to college if it were far but cheap would also go if it were close but expensive, or that the reverse is true.

I provide a simple approach to estimating causal effects under VM. In a setting with any number of binary instruments satisfying VM, I show that average treatment effects can be point identified for subgroups of the population if and only if that subgroup satisfies a certain condition.\footnote{In Appendix \ref{alternative}, I show how discrete instruments more generally can be accommodated by re-expressing them as a larger number of binary instruments, while preserving vector monotonicity and without loss of information.} The condition is met by, for example, the group of all units (e.g. individuals) that move into treatment when any fixed subset of the instruments are switched ``on''. As special cases, this includes for example the set of units that would respond to changing a single particular instrument, or those units for whom treatment status would vary in any way given changes to the available instruments. I propose a two-step estimator for this family of identified causal parameters.\footnote{This estimator is implemented in the companion Stata package \texttt{ivcombine}, available from \href{https://github.com/leonardgoff/ivcombine}{https://github.com/leonardgoff/ivcombine}.} Notably, the estimator has the same computational cost as the popular two-stage least squares (2SLS) estimator, despite the rapid proliferation of potential selection patterns compatible with VM as one increases the number of instruments.

Vector monotonicity represents a special case of what MTW refer to as \textit{partial monotonicity} (PM). VM and PM are very similar, but PM is ex-ante weaker: it allows the ``direction'' in which treatment uptake increases for each instrument to depend on the values of the other instruments. However given PM and the standard instrumental variables (IV) independence assumption, the additional restriction made by VM is testable. In particular, VM implies that the propensity score function is component-wise monotonic in the instruments. VM and PM thus coincide when this testable restriction is satisfied, and VM can be thought of as an application of PM within a class of settings that can be distinguished empirically. Further, VM is also often implied by natural choice-theoretic considerations, making monotonicity of the propensity score reasonable to expect provided that the instruments are valid.

In their paper, MTW focus on the causal interpretation of the 2SLS estimand under PM, and show that 2SLS is not guaranteed to recover a convex combination of heterogeneous treatment effects under PM or VM.\footnote{For example, Proposition 5 of MTW demonstrates this in the case of two binary instruments satisfying VM.} This motivates the question of what identifying power remains for instrumental variables satisfying PM or VM to uncover causal effects. In a second paper (\citealt*{Mogstad2020b}, henceforth MTW2), these same authors discuss identification more generally under partial monotonicity. MTW2 adapt the marginal treatment effects (MTE) framework of \citet{Heckman2005} for use under PM, and construct identified sets for a broad class of causal parameters that are typically only partially identified by IV methods (absent parametric assumptions and/or continuous instruments).

By contrast, my results maintain VM and characterize the class of causal parameters that are point identified without any auxiliary assumptions and even with discrete instruments. I show in Appendix \ref{sec:compare} that when VM does hold, the class of treatment effect parameters identified by my approach coincides with those of the same form that would be point identified under the approach of MTW2, if the method of MTW2 is applied using all identifying moments provided by the data but without additional maintained assumptions (e.g. parametric forms for MTEs).\footnote{My results thus also confirm a conjecture of MTW2---that their approach leads to identified sets that are sharp---in the setting I consider and when point identification holds.} In view of this, a desirable feature of my approach is that it is able to guarantee upfront to the researcher that their chosen target parameter is identified, and give a menu of such parameters that one could estimate. By leveraging constructive estimands for the target parameter, my results also lead to an easy-to-implement estimator and associated confidence intervals.

The estimator I propose in this paper can thus be seen as an alternative to the method of MTW2, but also to 2SLS, which is the most popular method to make use of multiple instruments in applied work. MTW derive additional testable conditions which are sufficient for the 2SLS estimand to  deliver positive weights under PM, but the number of conditions to be verified grows combinatorially with the number of instruments. Targeting a particular treatment effect parameter that is identified under VM avoids the need for such tests. My estimator couples this advantage with the computational ease of a simple ``2SLS-like'' estimator.

In Section \ref{setup}, I review the basic IV setup with a binary treatment, and compare VM to the traditional LATE monotonicity assumption and PM. In Section \ref{comono} I set the stage for the identification analysis by showing how with any number of binary instruments VM partitions the population into well-defined ``response groups'', nesting results from MTW for the two-instrument case. I then use this taxonomy of response groups in Section \ref{idsec} to characterize the family of identified parameters under VM with binary instruments, which leads to the estimator proposed in Section \ref{est}. Section \ref{empiricalapp} applies my method to study the labor market returns to college.

\section{Setup} \label{setup}

Suppose the researcher has a scalar outcome variable $Y$, a binary treatment variable $D$, and a vector $Z = (Z_1,Z_2, \dots ,Z_J)'$ of $J$ instrumental variables that can take values in $\mathcal{Z} \subseteq (\mathcal{Z}_1 \times \mathcal{Z}_2 \times \dots \times \mathcal{Z}_J)$, where $\mathcal{Z}_j$ denotes the set of values that instrument $Z_j$ can take.  A typical value in $\mathcal{Z}$ will be denoted with the boldface notation $\mathbf{z}$, with $z_j$ denoting the component corresponding to the $j^{th}$ instrument. I employ the standard definitions of potential outcomes and potential treatments, letting $D_i(\mathbf{z})$ denote the counterfactual treatment status of observational unit $i$ (e.g. an individual) when the vector of instruments takes value $\mathbf{z}$, and $Y_i(d,\mathbf{z})$ the outcome that would occur with treatment $d \in \{0,1\}$ and value $\mathbf{z}$ for the instruments. Let $Z_i=(Z_{1i},\dots, Z_{Ji})'$ denote unit $i$'s realized value of all $J$ instruments. 

The following assumption states that the $J$ available instrumental variables are \textit{valid}:
\begin{assumption*}[1 (exclusion \& independence)] $Y_i(d,\mathbf{z}) = Y_i(d)$ for all $\mathbf{z} \in \mathcal{Z}, d \in \{0,1\}$; and $(Y_i(1), Y_i(0), \{D_i(\mathbf{z})\}_{\mathbf{z} \in \mathcal{Z}}) \indep Z_i$.
\end{assumption*}
\noindent The first part of Assumption 1 states that the instruments satisfy the exclusion restriction that potential outcomes do not depend on instrument values once treatment status is fixed. The second part of Assumption 1 states that the instruments are statistically independent of potential outcomes and potential treatments.\footnote{It's worth noting that whether or not to use multiple instruments may not be ``optional'', in the sense that if a collection of instruments are valid, this does not imply that a subset of the instruments are as well.} In practice, it is common to maintain a version of this independence assumption that holds conditional on a set of observed covariates. I implicitly condition on any such covariates and discuss incorporating them in estimation in Appendix \ref{withcovariates}. 

\subsection{Notions of monotonicity} \label{sec:notions}

It is well-known that when treatment effects are heterogeneous, Assumption 1 alone is not sufficient for instrument variation to identify treatment effects. The seminal LATE model of \citet{Imbens2018} introduces the additional assumption of \textit{monotonicity}:
\begin{assumption*}[IAM (traditional LATE monotonicity)] For all $\mathbf{z}, \mathbf{z'} \in \mathcal{Z}$: $D_i(\mathbf{z}) \ge D_i(\mathbf{z}')$ for all $i$ or $D_i(\mathbf{z}) \le D_i(\mathbf{z}')$ for all $i$.
\end{assumption*}
\noindent I have referred to IAM as ``LATE monotonicity'' in the introduction, but for the remainder of the paper I follow MTW and call it IAM for short (for ``Imbens and Angrist monotonicity'').

To appreciate the sense in which IAM can be strong when $\mathbf{z}$ is a vector, let us code the two instruments for college from the introduction as binary variables (``far''/``close'' and ``cheap''/``expensive''). As emphasized by MTW, IAM says that a given counterfactual change to the proximity and/or tuition instruments can either move some students into college attendance, or some students out, but not both. In particular, this requires that all units who would go to college when it is far but cheap would also go to college if it was close and expensive, or the other way around. This implication will generally fail to hold if individuals differ in how much each of the instruments ``matters'' to them: for example, if some students are primarily sensitive to distance while others are primarily sensitive to tuition.\footnote{MTW also show that with continuous instruments, IAM implies the very strong restriction that marginal rates of substitution are identical among individuals indifferent between treatment and non-treatment.} 

Vector monotonicity instead captures monotonicity as the notion that increasing the value of any \textit{one} instrument weakly encourages (or discourages) all units to take treatment, regardless of the values of the other instruments.
\begin{assumption*}[2 (vector monotonicity)]
	There exists an ordering $\ge_j$ on $\mathcal{Z}_j$ for each $j \in \{1\dots J\}$ such that for all $\mathbf{z},\mathbf{z'} \in \mathcal{Z}$, if $\mathbf{z} \ge \mathbf{z}'$ component-wise according to the orderings $\{\ge_j\}_{j=1}^J$, then $D_i(\mathbf{z}) \ge D_i(\mathbf{z}')$ for all $i$.
\end{assumption*}
\noindent When each $\ge_j$ is the standard ordering on real numbers, MTE call VM ``actual monotonicity'', or AM.\footnote{\citet{Mountjoy2018} imposes a version of VM in a setting with continuous instruments and a ternary treatment.} I instead use the term ``VM'' to emphasize that $\ge_j$ need not be this order for identification results to hold, but I will typically restrict to AM (which represents a simple relabeling of the instrument values) for ease of exposition.

Assumption IAM implies the existence of a (total) order on $\mathcal{Z}$, where if $\mathbf{z} \ge \mathbf{z}'$ with respect to that order then $D_i(\mathbf{z}) \ge D_i(\mathbf{z'})$ for all $i$.\footnote{This follows since if $D_i(\mathbf{z}) \ge D_i(\mathbf{z}')$ and $D_i(\mathbf{z}') \ge D_i(\mathbf{z}'')$, then $D_i(\mathbf{z}) \ge D_i(\mathbf{z}'')$. Any two points in $\mathcal{Z}$ can be ranked in this way, yielding a weak total ordering on $\mathcal{Z}$.} In the returns-to-schooling example, this order might be the following, where an arrow from $\mathbf{z}'$ to $\mathbf{z}$ indicates that $D_i(\mathbf{z}) \ge D_i(\mathbf{z}')$ for all $i$: \vspace{.5cm} 
\tikzstyle{line} = [draw]
\begin{center}
	\begin{tikzpicture}
		\node (F1) at (1,0.5) {$(expensive, far)$};
		\node (F2) at (4.85,0.5) {$(cheap, far)$};
		\node (F3) at (9,0.5) {$(expensive, close)$};
		\node (F4) at (13,0.5) {$(cheap, close)$};
		
		\draw[->] (F1) -- (F2);
		\draw[->] (F2) -- (F3);
		\draw[->] (F3) -- (F4);
	\end{tikzpicture}
\end{center}
An alternative ordering to the one depicted above would be that instead $D_i(expensive, far) \le D_i(expensive, close) \le D_i(cheap, far) \le D_i(cheap, close)$. While either of these two orders may seem equally plausible ex-ante, Assumption IAM requires that only one or the other holds, common to all $i$ in the population.

By contrast, VM ascribes a \textit{partial} order on $\mathcal{Z}$---only some pairs $(\mathbf{z},\mathbf{z}')$ are ranked. In the returns to schooling example, the obvious partial order under VM is:
\begin{center}
	\begin{tikzpicture}
		
		\node (F1) at (1.5,-1.5) {$(expensive, far)$};
		\node (F2) at (6,-.5) {$(cheap, far)$};
		\node (F3) at (6,-2.5) {$(expensive, close)$};
		\node (F4) at (10.5,-1.5) {$(cheap, close)$};
		
		\draw[->] (F1) -- (F2);
		\draw[->] (F2) -- (F4);
		\draw[->] (F1) -- (F3);
		\draw[->] (F3) -- (F4);
	\end{tikzpicture}
\end{center}
The absence of vertical arrows between $(cheap, far)$ and $(expensive, close)$ above means that under VM, it could be the case that $D_i(cheap, far) > D_i(expensive, close)$ for some $i$, while $D_i(cheap, far) < D_i(expensive, close)$ for some other $i$.

The partial monotonicity assumption (PM) introduced by MTW is weaker than both IAM and VM. Like VM, it implies a partial order on $\mathcal{Z}$. Let $(z_j, \mathbf{z}_{-j})$ denote a given value in $\mathcal{Z}$ as the combination of a value $z_j \in \mathcal{Z}_j$ for the $j^{th}$ instrument and value $\mathbf{z}_{-j} \in \mathcal{Z}_{-j}$ for the other instruments, where $\mathcal{Z}_{-j}$ denotes the set of possible values for all instruments aside from $Z_j$.
\begin{assumption*}[\textbf{PM} (partial monotonicity)]
	Consider any $j\in\{1\dots J\}$, $z_j,z_j' \in \mathcal{Z}_{j}$, and  $\mathbf{z}_{-j} \in \mathcal{Z}_{-j}$ such that $(z_j, \mathbf{z}_{-j}) \in \mathcal{Z}$ and $(z'_j, \mathbf{z}_{-j}) \in \mathcal{Z}$. Then either $D_i(z_j,\mathbf{z}_{-j}) \ge D_i(z_j',\mathbf{z}_{-j})$ for all $i$ or $D_i(z_j,\mathbf{z}_{-j}) \le D_i(z_j',\mathbf{z}_{-j})$ for all $i$.
\end{assumption*}
\noindent Under PM, there exists for any instrument $j$ an ordering on the points $z \in \mathcal{Z}_j$ such that $D_i(z,\mathbf{z}_{-j})$ is weakly increasing along the order, \textit{for that fixed choice of} $\mathbf{z}_{-j}$. The key (and only) additional restriction made by VM beyond PM is that under VM, this ordering must be the same across all values of $\mathbf{z}_{-j}$ for a given $j$. For example, close proximity to a college encourages going to college, whether or not nearby colleges are cheap. By contrast, PM could capture a situation in which college proximity encourages attendance when nearby colleges are cheap, but discourages attendance when they are expensive.

\begin{figure}[h!] 
	\begin{center}
		\resizebox{.9\textwidth}{!}{
			\begin{tikzpicture}[clip=false]			
				\fill[gray, draw=black, thick, opacity=.4]   (150:0) circle (3.5);
				\node at ( 90:3.75)    {\textbf{PM}};				
				\fill[green, draw=black, thick, opacity=.4] (-1, -.5) circle (2);
				\node at (-1.1, 1.75)   {\textbf{VM}};
				\fill[blue, draw=black, thick, opacity=.4]  (1, -.5) circle (2);
				\node at (1,1.75)   {\textbf{IAM}};
				
				\begin{scope}[shift={(10,-.5)}]
					\pgfset{minimum size=4cm,inner sep=2mm}
					\color{green}
					\pgfsetstrokecolor{black}
					\pgfsetstrokeopacity{1}
					\pgfsetfillopacity{0.4}
					\pgfsetlinewidth{.25}
					\pgfnode{circle}{center}{}{nodename}{\pgfusepath{stroke,clip,fill}}
					
					\fill[blue, draw=black, thick, opacity=.4]  (2, 0) circle (2);
					
					\pgfsetfillopacity{0.2}
					\color{gray}
					\pgfnode{circle}{center}{}{nodename}{\pgfusepath{stroke,clip,fill}}				
				\end{scope}
				
				\node at (10, 1.85)   {\textbf{VM=PM}};
				\node at (11, -0.5)   {\textbf{IAM}};
				
				\node at (0, -4) {\textbf{With no restriction on propensity score}};
				\node at (10, -4) {\textbf{When propensity score is monotonic}};
				
			\end{tikzpicture}
		}
	\end{center}
	\caption{The left panel shows an ex-ante comparison of Imbens \& Angrist monotonicity (IAM), vector monotonicity (VM), and partial monotonicity (PM), if the propensity score function is unknown. The right panel depicts the relationship when the propensity score is component-wise monotonic: PM and VM become identical, with IAM a special case.} \label{venn}
\end{figure}

While VM is ex-ante stronger than PM, the additional restriction made by VM over PM is empirically testable, by inspecting the propensity score function $\mathcal{P}(\mathbf{z}):=\mathbbm{E}[D_i|Z_i=\mathbf{z}]$. Call $\mathcal{Z}$ \textit{non-disjoint} when for any two $\mathbf{z},\mathbf{z'} \in \mathcal{Z}$ there exists a sequence of vectors $\mathbf{z}_1, \dots, \mathbf{z}_M \in \mathcal{Z}$ where each $\mathbf{z}_m$ and $\mathbf{z}_{m-1}$ differ on only one component, and $\mathbf{z}_1=\mathbf{z}$, $\mathbf{z}_M=\mathbf{z}'$.\footnote{This property rules out atypical cases such as $\mathcal{Z}$ consisting only of the points $(0,0)$ and $(1,1)$ e.g. if $J=2$.}

\begin{proposition} \label{VMtest}
	Suppose PM and Assumption 1 hold, and $\mathcal{Z}$ is non-disjoint. Then VM holds if and only if $\mathcal{P}(\mathbf{z})$ is component-wise monotonic in $\mathbf{z}$, for some fixed ordering $\succeq_j$ on each $\mathcal{Z}_j$.
\end{proposition}
\begin{proof}
	Proofs for all results can be found in Appendix \ref{proofs} or the Online Appendix.
\end{proof}
\noindent Unlike VM, PM (like IAM) is compatible with any propensity score function $\mathcal{P}(\mathbf{z})$. Since IAM implies PM, it also follows from Proposition \ref{VMtest} that if IAM and Assumption 1 hold and $\mathcal{P}(\mathbf{z})$ is component-wise monotonic in $\mathbf{z}$, then VM holds. Thus if a researcher has verified that the propensity score function is monotonic,\footnote{With two binary instruments for example, one could test the four inequalities $\mathcal{P}(1,1) \ge \mathcal{P}(1,0)$, $\mathcal{P}(1,1) \ge \mathcal{P}(0,1)$, $\mathcal{P}(1,0) \ge \mathcal{P}(0,0)$, and $\mathcal{P}(0,1) \ge \mathcal{P}(0,0)$. This can be accomplished through a regression $D_i = \beta_0 + \beta_1 Z_{1i}+\beta_2 Z_{2i}+\beta_3 Z_{1i}Z_{2i}+\epsilon_i$ and testing that $\beta_1, \beta_2, \beta_3+\beta_1$ and $\beta_3 + \beta_2$ are all positive.} VM becomes a strictly weaker assumption than IAM. The overall relationship between Assumptions IAM, VM and PM is depicted in Figure \ref{venn}.\\

\noindent \textit{Remark 1:} Note that if Assumption 1 holds conditional on covariates $X_i$, Proposition \ref{VMtest} also need only hold with respect to the \textit{conditional} propensity score $\mathbbm{E}[D_i|Z_i=\mathbf{z},X_i=x]$ (see Section \ref{empiricalapp}). If VM is maintained, this property could in principle be used to test Assumption 1 conditional on a given set of covariates $X$.\\

\noindent \textit{Remark 2:} Another sufficient condition for VM given PM is the existence of individuals for each instrument that are responsive only to the value of that instrument. For example, suppose Alice only cares about proximity (going to college if and only if it is close), and Bob only cares about tuition (going to college if and only if it is cheap). If Alice and Bob are both present in the population, PM then requires that all other units in the population exhibit (weakly) the same directions of response to both instruments that Alice and Bob do, implying VM. The existence of both Alice and Bob in the population would also imply that IAM does not hold.

\section{Response groups under vector monotonicity} \label{comono}
To set the stage for analysis of identification under VM, I in this section show that VM partitions the population of interest into a set of groups that generalize the familiar taxonomy of ``always-takers'', ``never-takers'', and ``compliers'' from \citet{Imbens2018}, and also nests a taxonomy of six groups introduced by MTW for the case of two binary instruments.

To simplify notation, let $G_i$ represent an individual's entire vector of counterfactual treatments $\{D_i(\mathbf{z})\}_{\mathbf{z} \in \mathcal{Z}}$. For example, with a single binary instrument $G_i=\textit{always-taker}$ indicates that $D_i(0)=D_i(1)=1$. I refer to $G_i$ as unit $i$'s ``response group''.\footnote{This language follows \citet{Lee2020}. \citet{Heckman2018} use \textit{response-types} or \textit{strata}.} Response groups partition individuals in the population based on their selection behavior over all counterfactual values of the instruments. VM can be thought of as a restriction on the support $\mathcal{G}$ of $G_i$, limiting the types of response groups that can coexist in the population. 

While this section describes the structure of $\mathcal{G}$ under VM when the instruments are each binary,  Appendix \ref{alternative} shows how one can re-code a set of discrete but non-binary instruments as a larger number of binary instruments, while preserving VM. Also without loss of generality, let us let the value labeled ``1'' for each binary instrument be the direction in which potential treatments are increasing. These ``up'' values might be predicted ex-ante, but by Proposition \ref{VMtest} they are also empirically identified from the propensity score function.

\subsection{One or two binary instruments} \label{twoinstsec}
With one binary instrument, VM and IAM coincide. $\mathcal{G}$ then contains the three groups (see e.g. \citealt{Angrist2008}): ``compliers'' (for whom $D_i(1)>D_i(0)$), ``always-takers'' (for whom $D_i(1)=D_i(0)=1$) and ``never-takers'' (for whom $D_i(1)=D_i(0)=0$).

In the case of two binary instruments satisfying VM, MTW show that $\mathcal{G}$ contains six distinct response groups, enumerated in Table \ref{twobinarynames} below. In the returns to college setting, a ``$Z_1$ complier'', for example, would go to college if and only if college is cheap, regardless of whether it is close (like Bob).  A $Z_2$ complier, by contrast, would go to college if and only if college is close, regardless of whether it is cheap (like Alice). A reluctant complier requires college to \textit{both} be cheap and close to attend, while an eager complier goes to college so long as it is either cheap or close. Never and always takers are defined in the same way as they are under IAM: by $\max_{\mathbf{z} \in \mathcal{Z}} D_i(\mathbf{z})=0$ or $\min_{\mathbf{z} \in \mathcal{Z}} D_i(\mathbf{z})=1$.

\begin{table}[htb!]
	\centering
	\begin{tabular}{c|c|c|c|c|c|}
		\textbf{Name  }     & $\mathbf{z}=(0,0)'$ & $\mathbf{z}=(0,1)'$ & $\mathbf{z}=(1,0)'$ & $\mathbf{z}=(1,1)'$&$F$\\ \hline		never takers        & 0&0&0&0&n/a                                               \\
		always takers       & 1&1&1&1&$\emptyset$                                               \\
		$Z_1$ compliers     & 0&0&1&1&$\{1\}$                                               \\
		$Z_2$ compliers     & 0&1&0&1&$\{2\}$                                               \\
		reluctant compliers & 0&0&0&1&$\{1,2\}$                                               \\
		eager compliers     & 0&1&1&1&$\{1\},\{2\}$                        \\
	\end{tabular}\\ \vspace{.5cm}
	\caption{The six response groups under VM with two binary instruments, with names given by MTW. The first four columns of the table give $D_i(\mathbf{z})$, while the fifth yields an alternative representation of the response group, described in Section \ref{manybinary}.} \label{twobinarynames}
\end{table}

\subsection{Multiple binary instruments} \label{manybinary}
Now consider any number $J$ of binary instruments. For simplicity, let $\mathcal{Z} = \{0,1\}^J$, where $\{0,1\}^J = \{(z_1, z_2, \dots z_J): z_j \in \{0,1\}\}$ denotes the $J-$times Cartesian product of $\{0,1\}$.\footnote{If $\mathcal{Z}$ is a strict subset of $\{0,1\}^J$, the response groups can be defined from the restrictions to $\mathcal{Z}$ of the groups defined here. Appendix \ref{alternative} generalizes identification results to such cases, which arise with discrete instruments.} There are ex-ante $2^{|\mathcal{Z}|}=2^{2^{J}}$ distinct possible mappings between vectors of instrument values and treatment status, and we wish to characterize the subset of these that satisfy VM. The number of such response groups $G_i$ is the number of isotone boolean functions on $J$ variables, which is known to follow the so-called Dedekind sequence:\footnote{An analytical expression for the $\textrm{Ded}_J$ is given by \citet{A.Kisielewicz1988}, but only the first eight have been calculated numerically. While the Dedekind sequence explodes quite rapidly, it does so much more slowly than $2^{2^{J}}$ does. For example while $3/4=75\%$ of conceivable response groups for $J=1$ satisfy VM, only $20/256 \approx 7.8\%$ do for $J=3$, and just $ 7581/4294967296 \approx 1.7 *10^{-4}$ do for $J=5$. Thus the ``bite'' of VM is increasing with $J$, ruling out an increasing fraction of conceivable selection patterns.} $3, 6, 20, 168, 7581, 7828354, \dots $ \citep{A.Kisielewicz1988}. Letting $\mathrm{Ded}_J$ denote the $J^{th}$ value in this sequence, there are e.g. $\mathrm{Ded}_3=20$ response groups when there are three instruments, and $\mathrm{Ded}_4=168$ groups when there are four.

For an arbitrary $J$, we can enumerate these $\mathrm{Ded}_J$ groups as follows. One group that always satisfies VM is composed of ``never-takers'': those units for whom $D_i(\mathbf{z})=0$ for all values $\mathbf{z} \in \mathcal{Z}$. Each of the remaining response groups can be associated with a collection of sets of instruments. These sets represent minimal sets of the instruments that are sufficient for that unit to take treatment, if all instuments in the set take a value of one. For example, in a setting with three instruments, one response group would be the units that take treatment if either $Z_1=1$, or if $Z_2=Z_3=1$. We associate this response group with the collection of sets $\{1\}, \{2,3\}$. Note that by VM, any unit in this group must also take treatment if $Z_1=Z_2=Z_3=1$. Another response group might take treatment \textit{only} if $Z_1=Z_2=Z_3=1$, and is associated with the single set $\{1,2,3\}$. This response group is more ``reluctant'' than the former. The group of always-takers are the least ``reluctant'': they require no instruments to equal one in order for them to take treatment.

Formally, we can associate each response group aside from never-takers with a collection $F$ (which I refer to as a ``family'') of subsets $S \subseteq \{1\dots J\}$ of the instruments. A unit in the response group associated with family $F$ takes treatment when all instruments in any of the $S \in F$ are equal to one. However, we need only consider families for which no $S \in F$ is a subset of some other $S' \in F$. Families of sets having this property are referred to as \textit{Sperner families} (see e.g. \citealt{Kleitman1973}). Families that are not Sperner would be redundant under VM: for example, if $F$ consists of the set $\{2,3\}$ and $\{1,2,3\}$, then given VM the set $\{1,2,3\}$ could be dropped from $F$ without affecting the implied selection function $D_i(\mathbf{z})$.
\begin{definition}[(response group for a Sperner family)]
	For any Sperner family $F$, let $g(F)$ denote the response group described by the following treatment rule: $D_i(\mathbf{z})=\mathbbm{1}(z_j=1 \textrm{ for all } j \in S\textrm{, for at least one of the sets S in F})$.
\end{definition}
\noindent The response groups satisfying VM with $J$ binary instruments are thus: i) the never-takers group; and ii) $Ded_J-1$ further groups $g(F)$ corresponding to each distinct Sperner family (one such family is the null-set, which corresponds to always-takers).\footnote{For an explicit proof these exhaust all distinct $D:\{0,1\}^J \rightarrow \{0,1\}$, see e.g. \citet{Anderson1987} (Sec. 3.4.1).} 

In the simplest example when $J=1$, the only Sperner families are the null set and the singleton $\{1\}$: corresponding to always-takers and compliers, respectively. Together with never-takers, we have the familiar three groups from LATE analysis with a single binary instrument. For $J=2$, the five groups (apart from never takers) map to Sperner families shown in the rightmost column of Table \ref{twobinarynames}. With $J=3$ there are 19 Sperner families,\footnote{These are (listed each within bold brackets for legibility): $ \pmb{\{}\emptyset \pmb{\}}, \pmb{\{}\{1\}\pmb{\}}, \pmb{\{}\{2\}\pmb{\}}, \pmb{\{}\{3\}\pmb{\}}, \pmb{\{}\{1,2\}\pmb{\}}, \pmb{\{}\{1,3\}\pmb{\}},\pmb{\{}\{2,3\}\pmb{\}},$ $\pmb{\{}\{1,2,3\}\pmb{\}},\pmb{\{}\{1\},\{2\}\pmb{\}}, \pmb{\{}\{2\},\{3\}\pmb{\}}, \pmb{\{}\{1\},\{3\}\pmb{\}}, \pmb{\{}\{1\},\{2\},\{3\}\pmb{\}}, \pmb{\{}\{1,2\},\{3\}\pmb{\}}, \pmb{\{}\{1,3\},\{2\}\pmb{\}}, \pmb{\{}\{2,3\},\{1\}\pmb{\}},$\\$\pmb{\{}\{1,2\},\{1,3\}\pmb{\}}, \pmb{\{}\{1,2\},\{2,3\}\pmb{\}}, \pmb{\{}\{1,3\},\{2,3\}\pmb{\}}, \pmb{\{}\{1,2\},\{1,3\},\{2,3\}\pmb{\}}$.} An individual with $G_i = g\left(\{1,2\},\{1,3\},\{2,3\}\right)$, for instance, would take treatment so long as any two of the instruments take a value of one.

\subsection{``Simple'' response groups} \label{simple}
In a slight abuse of notation, let $D_g(\mathbf{z})$ the potential treatments function $D_i(\mathbf{z})$ that is common to all units sharing a value $g$ of $G_i$. A key difference between VM and IAM for identification is that under VM, the functions $D_g(\mathbf{z})$ for various response groups $g$ are not all linearly independent of one another. Indeed, as functions of $J$ binary variables, only $2^J$ such $D_g(\mathbf{z})$ could be independent, while $\mathrm{Ded}_J$ is strictly larger than $2^J$ for $J>1$. Let $\mathcal{G}^{c}:=\mathcal{G}/\{a.t., n.t.\}$ denote the set of $\mathrm{Ded}_J-2$ response groups aside from the never-takers and always takers that are compatible with VM. All of the groups in $\mathcal{G}^c$ can be thought of as generalized ``compliers'' of some kind: units that would vary treatment uptake in \textit{some} way in response to counterfactual changes to the values of the instruments.

We can construct a natural linear basis for the set of selection functions $\{D_g(\mathbf{z})\}_{g \in \mathcal{G}^c}$ by considering response groups $g(F)$ corresponding to Sperner families that consist of a single set $S$. I refer to such response groups, denoted $g(S)$, as \textit{simple}.\footnote{Note that a similar construction plays a central role in \citet{Lee2018a}.} For $J=2$, for example, we have:
$$D_{Z_1}(\mathbf{z})=z_1 \hspace{1cm} D_{Z_2}(\mathbf{z})=z_2\hspace{1cm}  D_{reluctant}(\mathbf{z})=z_1\cdot z_2$$
The selection function for the remaining group, eager compliers, can then be obtained as:
$$D_{eager}(\mathbf{z}) = z_1+z_2-z_1\cdot z_2= D_{Z_1}(\mathbf{z})+D_{Z_2}(\mathbf{z})-D_{reluctant}(\mathbf{z})$$
We can express this linear dependency across all groups by the matrix $M$ in the system:
\begin{equation} \label{m2sys}
	\begin{pmatrix}
		D_{Z_1}(\mathbf{z})\\D_{Z_2}(\mathbf{z})\\D_{reluctant}(\mathbf{z})\\D_{eager}(\mathbf{z})
	\end{pmatrix} = \underbrace{\begin{pmatrix}
			1 & 0 & 0\\
			0 & 1 & 0\\
			0 & 0 & 1\\
			1 & 1 & -1 
	\end{pmatrix}}_{:=M} \begin{pmatrix}
		D_{Z_1}(\mathbf{z})\\D_{Z_2}(\mathbf{z})\\D_{reluctant}(\mathbf{z})
\end{pmatrix} \end{equation}
Let $\mathcal{G}^s$ be the set of all simple response groups in $\mathcal{G}^c$. The set $\mathcal{G}^s$ is isomorphic to the collection of all subsets of $\{1 \dots J\}$ aside from the empty set (which corresponds to always-takers). For arbitrary $J$, we can define a $|\mathcal{G}^c| \times |\mathcal{G}^s|$ matrix $M$ that generalizes the linear system (\ref{m2sys}):
$$D_{g}(\mathbf{z}) = \sum_{g' \in \mathcal{G}^s} M_{gg'} \cdot D_{g'}(\mathbf{z}) \quad \textrm{ for all } g \in \mathcal{G}^c \textrm{ and } \mathbf{z} \in \mathcal{Z}$$
Let $F(g)$ denote the Sperner family corresponding to a given $g \in \mathcal{G}^c$ (i.e. the inverse of the function $g(G)$ in Definition 1). For any $g \in \mathcal{G}^s$, let $S(g)$ denote the lone set $S$ in $F(g)$. Given this notation, the entries of $M$ are given explicitly by the following expression:
\begin{proposition} \label{matrixlemma}
	$M_{gg'} = \sum_{F' \in s(F(g),S(g'))} (-1)^{|F'|+1}$ where $s(F,S') := \left\{F' \subseteq F: \left(\bigcup_{S \in F'} S\right) = S'\right\}$.
\end{proposition}
\noindent Fore completeness, note that for $g \in \mathcal{G}^s$, we have $D_{g}(\textbf{z}) = \left(\prod_{j \in S(g)}z_{j}\right) = \mathbbm{1}\left(z_j=1 \textrm{ for all }j \in S(g)\right)$.

\section{Identification under VM with binary instruments} \label{idsec}
In this section I define and characterize the class of causal parameters that are point identified under vector monotonicity, assuming that the instruments are binary and have full support.

Appendix \ref{alternative} shows that the assumption of binary instruments is without loss of generality in the sense that if one begins with finite discrete instruments satisfying vector monotonicity, these discrete instruments can be re-expressed as a larger number of binary instruments in a way that preserves VM. Appendix \ref{alternative} also relaxes the assumption of full-support, which is necessary to make use of this mapping. 

\subsection{Parameters of interest and identification} \label{secid}
To build up parameters of interest, I consider conditional averages of either potential outcome $Y_i(0)$ or $Y_i(1)$, after possible transformation by a function $f$. For $d\in \{0,1\}$, let
\begin{equation} \label{eq:thetaparam}
	\mu^{d}_c := \mathbbm{E}[f(Y_i(d))|C_i=1]
\end{equation}
where $C_i = c(G_i,Z_i)$ is any function $c: \mathcal{G} \times \mathcal{Z} \rightarrow \{0,1\}$ of individual $i$'s response group and their realization of the instruments. Intuitively, the event $C_i=1$ will indicate that unit $i$ belongs to a particular subgroup of generalized ``compliers''. Allowing $c$ to depend on $Z_i$ in addition to $G_i$ lets the practitioner focus attention on compliers that are responsive to some rather than all of the instruments, as described in Section \ref{secexamples}. Functions of the form $c(g,z)$ are the most general type of conditioning event that depends on the primitives of the IV model given in Section \ref{setup}, without depending directly on potential outcomes.\footnote{However, no restrictions are imposed on the joint distribution of $(Y_i(1),Y_i(0),G_i)$, so the model is compatible with $G_i$ being arbitrarily correlated with potential outcomes or with treatment effects, as in Roy-type models.} 

Most of the discussion will center on the class of average treatment effect parameters:
$$\Delta_c := \mathbbm{E}[Y_i(1)-Y_i(0)|C_i=1] = \mu^{1}_c - \mu^{0}_c$$
with $f(y)=y$ the identity function (for this reason I leave $f$ implicit in the notation $\mu^{d}_c$). The form $\Delta_c$ nests many treatment effect parameters familiar both from the LATE \citep{Imbens2018} and marginal treatment effects \citep{Heckman2005} literatures. For instance, with a single binary instrument the LATE sets $c(g,\mathbf{z})=\mathbbm{1}(g=complier)$. 

I now characterize the family of $c(g,\mathbf{z})$ under which identification of $\mu^{d}_c$ is possible. In particular, a necessary and sufficient condition will be what I call ``Property M'':
\begin{definition}[(Property M)]
	We say the function $c(g, \mathbf{z})$ satisfies Property M if both: 
	\begin{enumerate}[i)]
		\item $c(a.t.,\mathbf{z})=c(n.t.,\mathbf{z})=0$ for all $\mathbf{z} \in \mathcal{Z}$
		\item for every $g \in \mathcal{G}^c$ and $\mathbf{z} \in \mathcal{Z}$: $c(g,\mathbf{z})= \sum_{g' \in \mathcal{G}^s} M_{gg'}\cdot c(g',\mathbf{z})$
	\end{enumerate}
\end{definition}
\noindent I'll also say that a parameter $\mu^{d}_c$ or $\Delta_c$ ``satisfies Property M'' if its underlying function $c(g,\mathbf{z})$ does. Recall that the matrix $M$ is defined in Proposition \ref{matrixlemma}.

While Property M is somewhat abstract, the discussion in Section \ref{sec:intuition} will give intuition for its role in identification. Additionally, the following result connects Property M to the basic logic of of \citet{Imbens2018} underlying identification under IAM:
\begin{proposition} \label{propmsuff}
	A function $c: \mathcal{G} \times \mathcal{Z} \rightarrow \{0,1\}$ satisfies Property M if and only if
	$$c(g,\mathbf{z})= \sum_{k=1}^K \left\{D_g(\mathbf{u}_k(\mathbf{z}))-D_g(\mathbf{l}_k(\mathbf{z}))\right\}$$
	for some $K$, where $\mathbf{u}_k(\cdot)$ and $\mathbf{l}_k(\cdot)$ are functions $\mathcal{Z} \rightarrow \mathcal{Z}$ such that $\mathbf{u}_k(\mathbf{z}) \ge \mathbf{l}_k(\mathbf{z})$ component-wise while $\mathbf{l}_k(\mathbf{z}) \ge \mathbf{u}_{k+1}(\mathbf{z})$ component-wise, for all $k$ and $\mathbf{z} \in \mathcal{Z}$.
\end{proposition}
\noindent Proposition \ref{propmsuff} shows that average treatment effects that satisfy Property M can be written in the form $\Delta_c = \mathbbm{E}\left[Y_i(1)-Y_i(0)\left|i \in \bigcup_{k=1}^K \left\{i': D_{i'}(u_k(Z_i))>D_{i'}(l_k(Z_i))\right\}\right.\right]$. Specific examples are discussed in Section \ref{secexamples}. The restriction $\mathbf{u}_k(\mathbf{z}) \ge \mathbf{l}_k(\mathbf{z})$ implies that the expansion of $c(g,z)$ is into terms that each take a value of zero or one given VM, and $\mathbf{l}_k(\mathbf{z}) \ge \mathbf{u}_{k+1}(\mathbf{z})$ implies that only one of these can be equal to one for a given $(g,\mathbf{z})$. The proof of Proposition \ref{propmsuff} shows that we can also take $K \le J/2$ without loss of generality. 

As an example of a function $c$ that does not satisfy Property M, consider $c(g',\mathbf{z}) = \mathbbm{1}(g'=g)$, a function that picks out a single response group $g$. This cannot be written in the form of Proposition \ref{propmsuff} under VM when $J>1$, and we cannot in general identify the average treatment effect $\mathbbm{E}[Y_i(1)-Y_i(0)|G_i=g]$ within single response groups $g$.\footnote{We can see this in a simple example with $J=2$ and $g$ being a $Z_1$ complier. In this case Property M would require that $c(\textrm{eager},\mathbf{z}) = c(Z_1,\mathbf{z})+c(Z_2 ,\mathbf{z})-c(\textrm{reluctant},\mathbf{z})$, i.e. that $0=1+0-0$, by Eq. (\ref{m2sys}).} Another example is the ATE, as $c(g,\mathbf{z})=1$ for all $g \in \mathcal{G}$ including always- and never-takers violates item i) of Property M.

Causal parameters that satisfy Property M are identified under VM with binary instruments, provided the various instruments provide sufficient independent variation in treatment uptake. A simple sufficient condition for this is that the instruments have full (rectangular) support. This assumption is stronger than necessary (Appendix \ref{alternative} gives a generalization), but simplifies presentation. Let $\mathbb{S}_Z:=\{\mathbf{z} \in \mathcal{Z}: P(Z_i=\mathbf{z})>0\}$ be the support of the random variable $Z_i$.\footnote{I distinguish between $\mathbb{S}_Z$ and the set of conceivable values $\mathcal{Z}$ because some results (e.g. Proposition \ref{VMtest}) can leverage the assumption that VM holds for values $\mathbf{z} \in \mathcal{Z}$ even if they have zero probability in the population.}
\begin{assumption*} [3 (full support)] 
	$\mathbb{S}_Z=\{0,1\}^J$ .
\end{assumption*}

An alternative expression of Assumption 3 is useful for stating the constructive identification result below. For an arbitrary ordering $g_1 \dots g_k$ of the $k:=2^J-1$ simple response groups in $\mathcal{G}^s$, define a $k$-component random vector $\Gamma_i=(D_{g_1}(Z_i), \dots, D_{g_k}(Z_i))'$ where each component gives the treatment status for a particular response group given realization $Z_i$ of the instruments.\footnote{\label{fn:equivalently}Equivalently, $\Gamma_i= (Z_{S_1i} \dots, Z_{S_ki})'$ for some arbitrary ordering of the $k=2^J-1$ non-empty subsets $S \subseteq \{1\dots J\}$, where $Z_{Si} := \prod_{j \in S} Z_{ji}$ and $g_\ell = g(S_\ell)$ for $\ell=1 \dots k$.} Let $\Sigma$ be the $k \times k$ variance-covariance matrix of $\Gamma_i$.
\begin{lemma} \label{lemmaequiv}
	Assumption 3 holds if and only if $\Sigma$ has full rank.
\end{lemma}
\noindent Lemma \ref{lemmaequiv} demonstrates that full support of the instruments is equivalent to there being linearly independent variation in treatment take-up among all of the simple response groups.

Theorem \ref{thmid} provides an explicit estimand for $\mu^{d}_c$ when the function $c$ satisfies Property M:
\begin{theorem} \label{thmid}
	Under Assumptions 1-3 (independence \& exclusion, VM, and full support), for any $c$ satisfying Property M and any measurable function $f(\cdot)$: 
	$$P(C_i=1)=E[h(Z_i)D_i] \quad \quad \textrm{ and } \quad \quad \mu^{d}_c = (-1)^{d+1} \frac{\mathbbm{E}[f(Y_i)h(Z_i)\mathbbm{1}(D_i=d)]}{\mathbbm{E}[h(Z_i)D_i]},$$ \vspace{.1cm}
	provided that $P(C_i=1)>0$, where $h(Z_i) = \mathbf{\lambda}'\Sigma^{-1}(\Gamma_i - \mathbbm{E}[\Gamma_i])$ and $$\mathbf{\lambda} =(\mathbbm{E}[c(g_1,Z_i)],\mathbbm{E}[c(g_2,Z_i)], \dots \mathbbm{E}[c(g_k,Z_i)])'$$
\end{theorem}
\noindent It follows immediately from Theorem \ref{thmid} that conditional average treatment effects $\Delta_c=\mu^{1}_c - \mu^{0}_c$ satisfying Property M are identified, and the expression simplifies to: $\Delta_c = \mathbbm{E}[h(Z_i)Y_i]/\mathbbm{E}[h(Z_i)D_i]$ (using that $ \mathbbm{1}(D_i=0)+\mathbbm{1}(D_i=1)=1$). Note that as the numerator of $\Delta_c$ depends on $Z_i$ and $Y_i$ only and the denominator depends on $Z_i$ and $D_i$ only, identification of $\Delta_c$ would hold in a ``split-sample'' setting where $Y_i$ and $D_i$ are not necessarily known for the same individual.

Now I show that Theorem \ref{thmid} has a converse: \textit{any} identified $\Delta_c$ must satisfy Property M. In this sense, Property M is both necessary and sufficient for identification. To state this result, let us consider so-called ``IV-like estimands'' introduced by \citet{Mogstad2018}, which are any cross moment $\mathbbm{E}[s(D_i,Z_i)Y_i]$ between $Y_i$ and a function of treatment and instruments for some function $s$. Let $\mathcal{P}_{DZ}$ denote the joint distribution of $D_i$ and $Z_i$, which is identified. Then:
\begin{theorem} \label{necctheorem}
	Suppose $\mu^{d}_c$ is identified for each $d\in \{0,1\}$ by a finite set of IV-like estimands and $\mathcal{P}_{DZ}$, given Assumptions 1-3 and $P(C_i=1)>0$. Then $\mu^{d}_c$ satisfies Property M.
\end{theorem}
\noindent Since the identification approach of MTW2 relies on IV-like estimands for identification, Theorem \ref{necctheorem} implies that any parameter of the form $\Delta_c$ that is identified by MTW2's approach is also identified by Theorem \ref{thmid} (provided no additional restrictions are leveraged with MTW2's approach).\footnote{In saying that a parameter $\theta$ is \textit{identified} by a particular set of empirical estimands, I mean that the set of values of $\theta$ that are compatible with the empirical estimands and maintained assumptions is a singleton, for any joint distribution of the model primitives---in this case $(G_i,Y_i(1), Y_i(0),Z_i)$---that is compatible with those assumptions \citep{idzoo}.\label{fn:identified}} In Appendix \ref{sec:compare} I show that the reverse is also true: when MTW2's approach to identification is leveraged with a ``complete'' set of IV-like estimands, it also point identifies all parameters $\Delta_c$ that satisfy Property M.

\subsection{An algebraic intuition for Theorem \ref{thmid}} \label{sec:intuition}

Before turning to examples, this section provides an algebraic intuition for Theorem \ref{thmid}. For simplicity, I focus on average treatment effect parameters $\Delta_c$.

By Assumption 1 and the law of iterated expectations, we can write any $\Delta_c$ as a weighted average over response-group specific average treatment effects $\Delta_g:=\mathbbm{E}[Y_i(1)-Y_i(0)|G_i=g]$:
\begin{align}
	\Delta_c = \sum_{g \in \mathcal{G}}\left\{\frac{P(G_i=g)\cdot \mathbbm{E}[c(g,Z_i)]}{\sum_{g' \in \mathcal{G}}P(G_i=g')\cdot \mathbbm{E}[c(g',Z_i)]}\right\}\cdot \Delta_g \label{deltaw}
\end{align}
where notice that the weight on $\Delta_g$ is proportional to the quantity $\mathbbm{E}[c(g,Z_i)]$, as well as $P(G_i=g)$. Now consider a general type of IV estimand in which a single scalar $h(Z_i)$ is constructed from the vector of instruments $Z_i$ according to a function $h$, and then used as a single ``instrument'' in linear IV regression.\footnote{Special cases of this form include 2SLS: $h(\mathbf{z}) = \mathcal{P}(\mathbf{z})$, and Wald-type estimands: $h(\mathbf{z}) = \frac{\mathbbm{1}(Z_i=\mathbf{z})}{P(Z_i=\mathbf{z})}-\frac{\mathbbm{1}(Z_i=\mathbf{z'})}{P(Z_i=\mathbf{z'})}$.}
Some algebra shows that under Assumption 1:
\begin{equation}
	\frac{Cov(Y_i,h(Z_i))}{Cov(D_i,h(Z_i))} = \sum_{g \in \mathcal{G}} \frac{P(G_i=g)\cdot Cov(D_g(Z_i), h(Z_i))}{\sum_{g' \in \mathcal{G}}P(G_i=g')\cdot Cov(D_{g'}(Z_i), h(Z_i))} \cdot \Delta_g \label{2slslike}
\end{equation}
These estimands therefore also uncover a weighted average of the $\Delta_g$, similar to (\ref{deltaw}). In (\ref{2slslike}) however, the weight placed on each response group $g$ is governed by the covariance between $D_g(Z_i)$ and $h(Z_i)$. Thus a simple IV estimand using $h(Z_i)$ can identify $\Delta_c$ if the function $h$ is chosen in such a way that $Cov(D_{g}(Z_i), h(Z_i))=\mathbbm{E}[c(g,Z_i)]$ for each of the response groups $g$. Since the covariance operator is linear, the linear dependencies among the functions $D_g(\cdot)$ captured by the matrix $M$ in Section \ref{manybinary} translate into a linear restrictions that must hold among the $\mathbbm{E}[c(g,Z_i)]$. Property M guarantees that the $\mathbbm{E}[c(g,Z_i)]$ satisfy these restrictions, regardless of the distribution of $Z_i$. What remains is then to simply ``tune'' the covariances $Cov(D_{g}(Z_i), h(Z_i))$ for each simple response group $g \in \mathcal{G}^s$ by careful choice of $h(\cdot)$. This is possible when the instruments have full support via the function $h(\cdot)$ defined in Theorem \ref{thmid}. A direct proof of Theorem \ref{thmid} along these lines is provided in the Online Appendix. The main proof provided in Appendix \ref{proofs} is more involved, and is structured around building a foundation for the comparison to the identification approach of MTW2 in Appendix \ref{sec:compare}.

The need for Property M in Theorem \ref{thmid} thus arises from there being under VM more response groups in $\mathcal{G}^c$ than there are independent pairs of points in the support of the instruments. This contrasts with IAM, under which both are generally equal (with binary instruments) to $2^J-1$.\footnote{Under IAM, there is an order on the $2^J$ points in $\mathcal{Z}$ such that between any two adjacent instrument values $\mathbf{z}, \mathbf{z}'$ along that order, there is a type of complier $g$ that first takes treatment at $\mathbf{z}$, and $\Delta_g = \frac{\mathbbm{E}[Y_i|Z_i=\mathbf{z}']-\mathbbm{E}[Y_i|Z_i=\mathbf{z}]}{\mathbbm{E}[D_i|Z_i=\mathbf{z}']-\mathbbm{E}[D_i|Z_i=\mathbf{z}]}$.} As a result, it is possible under IAM to identify $\Delta_{g}$ for any single such response group $g \in \mathcal{G}^c$. However, under VM the corresponding choice $c(g',\mathbf{z}) = \mathbbm{1}(g'=g)$ fails to satisfy Property M, as described in Section \ref{secid}.

\subsection{Some examples from the family of identified parameters} \label{secexamples}
This section highlights some of parameters $\Delta_c$ that are identified under VM according to Theorem \ref{thmid}, and discusses their interpretation in the returns to schooling setting mentioned in the introduction. Let $\Delta_i = Y_i(1)-Y_i(0)$ be the treatment effect for unit $i$, and $\mathcal{J} \subseteq \{1, \dots J\}$ be any subset of the instruments. Proposition \ref{propmsuff} shows that each of the parameters introduced in Table \ref{tablelates} below satisfy Property M when $\mathcal{Z} = \{0,1\}^J$, and are hence identified by Theorem \ref{thmid}.\footnote{\label{fn:closure} Some further examples of identified parameters from those mentioned in Table \ref{tablelates} can be constructed using a closure property of the set of $c$ satisfying Property M. Let $\mathcal{C}$ denote the set of $c: \mathcal{G} \times \mathcal{Z} \rightarrow \{0,1\}$ that satisfy Property M, and let $c_a(g,\mathbf{z})$ and $c_b(g,\mathbf{z})$ be two functions in $\mathcal{C}$. Then $c_a(g,\mathbf{z})-c_b(g,\mathbf{z}) \in \mathcal{C}$ iff $c_b(g,\mathbf{z}) \le c_a(g,\mathbf{z})$ for all $\mathbf{z} \in \mathcal{Z},g \in \mathcal{G}^c$. We can use this observation to generate identified parameters that condition on the complement of the complier group for ${c_b}$ within the larger complier group for ${c_a}$. For example with $J=2$, consider the average treatment effect among individuals who are counted in the ACLATE but not in $SLATE_{\{1\}}$:
	$\mathbbm{E}[\Delta_i|G_i \in \mathcal{G}^c \textrm{ but } \left\{D_i(1,Z_{2i}) = D_i(0,Z_{2i})\right\}]$. This represents the average effect among individuals that would not respond to reduction in college tuition alone, but would respond if both tuition and proximity were shifted in concert.}

\begin{table}[bth!]
	\centering 
	\small
	\begin{tabular}{Hl|c|Hcc}
		&\textbf{Name}&\textbf{Definition}&& \textbf{Prop.} \ref{propmsuff} \textbf{form of} $\mathbf{c(g,z)}$\\ \hline
		1.& $ACLATE$ & $\mathbbm{E}[\Delta_i|G_i \in \mathcal{G}^c]$& $\mathbbm{1}(g \in \mathcal{G}^c)$& $D_g(1\dots 1)-D_g(0\dots0)$ \\
		2. & $ SLATE_{\mathcal{J}}$ & $\mathbbm{E}[\Delta_i|D_i(1\dots1,Z_{-\mathcal{J},i})>D_i(0\dots0,Z_{-\mathcal{J},i})]$& $\mathbbm{E}[\Delta_i|D_i(1\dots1,\mathbf{z}_{-\mathcal{J}})>D_i(0\dots0,\mathbf{z}_{-\mathcal{J}})]$& $D_g(1\dots 1,\mathbf{z}_{-\mathcal{J}})-D_g(0\dots0,\mathbf{z}_{-\mathcal{J}}))$\\
		3. & $ SLATT_{\mathcal{J}}$ &  $\mathbbm{E}[\Delta_i|D_i(1\dots1,Z_{-\mathcal{J},i})>D_i(0\dots0,Z_{-\mathcal{J},i}),D_i=1]$&$D_g(z)\cdot \left(D_g(1\dots1,\mathbf{z}_{-\mathcal{J}})-D_g(0\dots0,\mathbf{z}_{-\mathcal{J}}))\right)$&$D_g(z)-D_g(0\dots0,\mathbf{z}_{-\mathcal{J}})$\\
		4. & $ SLATU_{\mathcal{J}}$& $\mathbbm{E}[\Delta_i|D_i(1\dots1,Z_{-\mathcal{J},i})>D_i(0\dots0,Z_{-\mathcal{J},i}),D_i=0]$& $(1-D_g(z))\cdot \left(D_g((1\dots1),\mathbf{z}_{-\mathcal{J}})-D_g(0\dots0,\mathbf{z}_{-\mathcal{J}}))\right)$&$D_g(1\dots1,\mathbf{z}_{-\mathcal{J}})-D_g(z)$\\
		5. &$PTE_j(\mathbf{z}_{-j})$ & $\mathbbm{E}[\Delta_i|D_i(1,\mathbf{z}_{-j})>D_i(0,\mathbf{z}_{-j})]$&& $D_g(1,\mathbf{z}_{-j})-D_g(0,\mathbf{z}_{-j}))$\\
	\end{tabular} \\ \vspace{.5cm}
	\caption{Examples of identified treatment effect parameters under VM (see text for details). Here $Z_{-\mathcal{J},i}$ denotes the components $Z_{ji}$ of $Z_i$ for $j \notin \mathcal{J}$, and $(d\dots d,Z_{-\mathcal{J},i})$ denotes a vector in which the remaining components $Z_j$ for $j \in \mathcal{J}$ are all equal to $d$. \label{tablelates}}
\end{table}
\noindent I call the first item in Table \ref{tablelates} the ``all-compliers LATE'' (ACLATE). The ACLATE is the average treatment effect among all units who would change their treatment uptake in any way in response to the instruments, and is the largest subgroup of the population for which treatment effects can be generally point identified from instrument variation alone. In the returns to schooling example, the ACLATE can be described as the average treatment effect among individuals who would go to college were it close and cheap, but would not were it far and expensive. 

A \textit{set local average treatment effect}, or $SLATE_\mathcal{J}$, captures the average treatment effect among units that move into treatment when all instruments in some fixed set $\mathcal{J}$ are changed from 0 to 1, with the other instruments not in $\mathcal{J}$ remaining at that unit's realized values. The ACLATE is a special case of SLATE when $\mathcal{J} = \{1, \dots J\}$. In the other extreme where $\mathcal{J}$ contains just one instrument index, SLATE recovers treatment effects among those who would ``comply'' with variation in that instrument alone. For example, $SLATE_{\{2\}}$ is the average treatment effect among individuals who don't go to college if it is far, but do if it is close (given their realized value of the tuition instrument).\footnote{\label{fncontrol }Note that a single-instrument SLATE like $SLATE_{\{2\}}$ does \textit{not} generally correspond to using $Z_2$ alone as an instrument, e.g. $Cov(Y,Z_2)/Cov(D,Z_2)$, unless $Z_1$ and $Z_2$ are independent.} This parameter may for example be of interest to policymakers considering whether to expand a community college to a new campus, and is related to the marginal treatment effect curve for instrument $j$ (see Appendix \ref{sec:compare}).

The treatment effect parameters $SLATT_{\mathcal{J}}$ and $SLATU_{\mathcal{J}}$ are similar to $SLATE_{\mathcal{J}}$ but additionally condition on units' realized treatment status. For example $SLATT_{\{1,2\}}$ with our two instruments averages over individuals who do go to college, but wouldn't have gone were it far and expensive.\footnote{Note that with a single binary instrument, $SLATT_{\{1\}}$ coincides with $ACLATE=SLATE_{\{1\}}$, as $\mathbbm{E}[\Delta_i|D_i=1, G_i = complier]=\mathbbm{E}[\Delta_i|Z_i=1, complier] = \mathbbm{E}[\Delta_i|complier]$, using Assumption 1. However, when the group $\mathcal{G}^c$ consists of more than one group, the ``all-compliers'' version of $SLATT$ generally differs from $ACLATE$.} The final row of Table \ref{tablelates} gives the most disaggregated type of identified parameter that can be identified under VM, what might be called a \textit{partial treatment effect} $PTE_j(\mathbf{z}_{-j})$. This is the average treatment effect among individuals that move into treatment when a single instrument $j$ is shifted from zero to one, while the other instrument values are held fixed at some explicit vector of values $\mathbf{z}_{-j}$. An example is the average treatment effect among individuals who go to college if it is close and cheap, but do not if it is far and cheap. 

Section \ref{est} discusses estimation of the parameters listed in Table \ref{tablelates}. The next section first outlines some further remarks on identification under VM.

\subsection{Extensions and further results on identification} \label{secdiscussion}

\textit{1) Linear dependency among the instruments.} Assumption 3 is stronger than is strictly necessary for identification, since linear dependencies between products of the instruments pose no problem if the corresponding ``weights'' in $\Delta_c$ do not need be tuned independently from one another. In Appendix \ref{alternative}, I give a version of Assumption 3 and generalization of Theorem \ref{thmid} that can accommodate instrument support restrictions and instruments that are not binary.

\textit{2) Conditional distributions of the potential outcomes.} By choosing $f(Y) = \mathbbm{1}(Y \le y)$ for a value $y$ in the support of $Y_i$, we can through Theorem \ref{thmid} identify the CDF of each potential outcome at $y$ conditional on $C_i=1$ (provided that $(Y_i,Z_i,D_i)$ are all observed in the same sample). This allows for the identification of quantile treatment effects or bounds on the distribution of treatment effects \citep{Theory2019}, in either case conditional on $C_i=1$.

\textit{3) Identified sets for ATE, ATT, and ATU.} When $Y_i$ has bounded support, we can generate sharp bounds in the spirit of \citet{Manski2019} for parameters like the average treatment effect (ATE), using the identified parameters in this paper. For example, the ATE can be written as $ATE := \mathbbm{E}[Y_i(1)-Y_i(0)] = p_a \Delta_{a} + p_n \Delta_{n} + (1-p_n-p_a)ACLATE$,
where $\Delta_a=\mathbbm{E}[Y_i(1)-Y_i(0)|G_i=a.t.]$, and $p_a=P(G_i=a.t.)$ (and analogously for $\Delta_n$ and $p_n$). Both $p_a$ and $p_n$ are point identified, while bounds on $\Delta_n$ and $\Delta_a$ can be obtained from the support of $Y_i$. The SLATT and SLATU can similarly be used to construct bounds on the average treatment effect on the treated (ATT) or untreated (ATU).

\section{Estimation} \label{est}
This section proposes a simple two-step estimator for the family of identified causal parameters introduced in Section \ref{idsec}, focusing on the conditional average treatment effects $\Delta_c$. The estimator is asymptotically normal and converges at the parametric rate.

Theorem \ref{thmid} establishes that a $\Delta_c$ satisfying Property M is identified by a ratio of two population expectations $\mathbbm{E}[h(Z_i)Y_i]/\mathbbm{E}[h(Z_i)D_i]$, so a natural estimator simply replaces these expectations with their sample counterparts, plugging in a first-step estimate of $h(Z_i)$. Recall that the function $h(\cdot)$ is defined from the vector $\mathbf{\lambda} =(\mathbbm{E}[c(g_1,Z_i)], \dots \mathbbm{E}[c(g_k,Z_i)])'$ where $k=2^J-1$. Given an $i.i.d.$ sample of size $n$ and a consistent estimator $\hat{\mathbf{\lambda}}$ of $\mathbf{\lambda}$, we can estimate $\Delta_c$ by $\hat{\rho}(\hat{\mathbf{\lambda}})$, where
\begin{equation} \label{rhonoreg}
	\hat{\rho}(\mathbf{\lambda}) := \left((0,\mathbf{\lambda}')(\Gamma'\Gamma)^{-1}\Gamma'D\right)^{-1}(0,\mathbf{\lambda}')(\Gamma'\Gamma)^{-1}\Gamma'Y
\end{equation}
and we introduce a $n \times 2^J$ matrix $\Gamma$ comprised of rows $(0,\Gamma_i')$ for each observation $i$, as well as $n \times 1$ vectors $D$ and $Y$ comprised of observations of $D_i$ and $Y_i$.\footnote{To obtain (\ref{rhonoreg}), recall that $h(Z_i) = (\Gamma_i - \mathbbm{E}[\Gamma_i])'\mathbbm{E}[(\Gamma_i - \mathbbm{E}[\Gamma_i])(\Gamma_i - \mathbbm{E}[\Gamma_i])']^{-1} \mathbf{\lambda}$. Accordingly, let $\hat{H} =  n \tilde{\Gamma}(\tilde{\Gamma}'\tilde{\Gamma})^{-1}\hat{\mathbf{\lambda}}$ be a vector $\hat{H}$ of estimates for $h(Z_i)$, where $\tilde{\Gamma}$ is a $n \times k$ matrix with entries $\tilde{\Gamma}_{il} = D_{g_\ell}(Z_i) - \frac{1}{n}\sum_{j=1}^n D_{g_\ell}(Z_j)$, and $g_\ell$ is the $\ell^{th}$ response group for some arbitrary ordering of the $k:=2^J-1$ groups $g_\ell \in \mathcal{G}^s$. Now consider $\hat{\Delta}_c = (\hat{H}'D)^{-1}(\hat{H}'Y)$, where $Y$ and $D$ are $n \times 1$ vectors of observations of $Y_i$ and $D_i$, respectively. By the Frisch-Waugh-Lovell theorem, $(\tilde{\Gamma}'\tilde{\Gamma})^{-1}\tilde{\Gamma}'D$ is the same as the final $k$ components of the vector $(\Gamma'\Gamma)^{-1}\Gamma'D$. Thus $\mathbf{\lambda}'(\tilde{\Gamma}'\tilde{\Gamma})^{-1}\tilde{\Gamma}'D = (0,\mathbf{\lambda}')(\Gamma'\Gamma)^{-1}\Gamma'D$, and similarly for $Y$.} Note that the population analog of $(\Gamma'\Gamma)^{-1}$ exists by Assumption 3. However the RHS of (\ref{rhonoreg}) is still consistent for $\Delta_c$ when Assumption 3 is relaxed as in Appendix \ref{alternative} (with $\lambda$ and $\Gamma$ modified as described therein).

Table \ref{tablelatesest} below gives examples of $\hat{\mathbf{\lambda}}$ for leading treatment effect parameters. With $\hat{\lambda} \stackrel{p}{\rightarrow} \lambda$ in each of these examples, we have that $\hat{\rho}(\hat{\mathbf{\lambda}}) \stackrel{p}{\rightarrow} \sum_{g \in \mathcal{G}^c}\frac{P(G_i=g)[M\mathbf{\lambda}]_g}{\sum_{g' \in \mathcal{G}^c}P(G_i=g')[M\mathbf{\lambda}]_{g'}}\cdot \Delta_g$ under standard regularity conditions. Matching this estimand to particular parameters $\Delta_c$ that satisfy Property M is achieved by choosing $\hat{\mathbf{\lambda}}$ appropriately for that $\Delta_c$. Asymptotic normality $\hat{\rho}(\hat{\mathbf{\lambda}})$ follows as a special case of Theorem 3 in \citet{Imbens2018}, which provides an expression of the estimator's asymptotic variance. The estimator $\hat{\Delta}_c=\hat{\rho}(\hat{\mathbf{\lambda}})$ and accompanying confidence intervals for $\Delta_c$ are implemented in the Stata package \texttt{ivcombine}.\\

\begin{table}[bth!]
	\centering 
	\begin{tabular}{Hl|c}
		&\textbf{Parameter}&Estimator $\hat{\mathbf{\lambda}}$ of population $\mathbf{\lambda}$\\
		\hline 
		1. & $ACLATE$ & $(1,1,\dots 1)'$ \\
		2. & $ SLATE_{\mathcal{J}}$ & $\hat{\mathbf{\lambda}}_S = \mathbbm{1}(\mathcal{J} \cap S \ne \emptyset)\hat{P}\left(Z_{S-\mathcal{J},i}=1\right)$\\
		3. & $ SLATT_{\mathcal{J}}$ & $\hat{\mathbf{\lambda}}_S = \mathbbm{1}(\mathcal{J} \cap S \ne \emptyset)\hat{P}\left(Z_{S,i}=1\right)$\\
		4. & $ SLATU_{\mathcal{J}}$ & $\hat{\mathbf{\lambda}}_S = \mathbbm{1}(\mathcal{J} \cap S \ne \emptyset)\hat{P}\left(Z_{S-\mathcal{J},i}(1-Z_{\mathcal{J},i})=1\right)$\\
		5. &$PTE_j(\mathbf{z}_{-j})$ & $\hat{\mathbf{\lambda}}_S = \mathbbm{1}(S=\mathbf{z}_{-j,1} \cup \{j\})$
	\end{tabular} \\ \vspace{.5cm}
	\caption{Leading examples of $\hat{\mathbf{\lambda}}$, where to ease notation I index element $g \in \mathcal{G}^c$ of $\hat{\lambda}$ by its associated set $S=S(g)$. Here $Z_{S,i} := \prod_{j \in S} Z_{ji}$, $S-\mathcal{J}$ denotes the set difference $\{j: j \in S, j \notin \mathcal{J}\}$, $\mathbf{z}_{-j,1}$ denotes the set of instruments that are equal to one in $\mathbf{z}_{-j}$, and $\hat{P}(E_i) = n^{-1} \sum_{i=1}^n \mathbbm{1}(E_i)$ for any event $E_i$.} \label{tablelatesest}
\end{table}

\noindent \textit{Comparison with 2SLS:} The estimator $\hat{\Delta}_c$ has a similar form to a ``fully-saturated'' 2SLS estimator that includes an indicator for each value of $Z_i$ in the first stage. Indeed, this version of 2SLS can be written in the form $\hat{\rho}(\mathbf{\lambda})$ where the components of $\mathbf{\lambda}$ are sample covariances between $D_i$ and a given component of $\Gamma_i$, corresponding to the estimand: $\rho_{2sls}= \sum_{g \in \mathcal{G}^c} \frac{ P(G_i=g)\cdot Cov(D_i,D_g(Z_i))}{\sum_{g'}  P(G_i=g')\cdot Cov(D_i,D_{g'}(Z_i))}\cdot \Delta_g$. The weights that 2SLS uses to aggregate over linear projection coefficients $(\Gamma'\Gamma)^{-1}\Gamma'D$ and $(\Gamma'\Gamma)^{-1}\Gamma'Y$ are thus determined asymptotically by the joint distribution of $D_i$ and $Z_i$, which the researcher has no control over. MTW show that the implied weight on some $\Delta_g$ under VM may be negative, depending on the DGP. By contrast, $\hat{\Delta}_c$ uses $\hat{\lambda}$ chosen to match the desired parameter of interest, guaranteeing that the estimator recovers a well-defined causal parameter under VM. Even with a large number of instruments, $\hat{\Delta}_c$ is no more ``expensive'' than 2SLS: both involve computing two linear projections each with $2^J$ of terms (despite the fact that number of underlying selection groups is much larger under VM compared with IAM).\\ 

\noindent \textit{Estimation of the ACLATE from a single Wald ratio:} The population estimand corresponding to the all-compliers LATE takes on a particularly simple form, a single ``Wald ratio'':
\begin{equation} \label{waldeq}
ACLATE := \frac{\mathbbm{E}[Y_i|Z_i=\bar{Z}] - \mathbbm{E}[Y_i|Z_i=\underbar{Z}]}{\mathbbm{E}[D_i|Z_i=\bar{Z}] - \mathbbm{E}[D_i|Z_i=\underbar{Z}]}
\end{equation}
where $\bar{Z} = (1\dots1)'$ and $\underbar{Z} = (0\dots0)'$, provided that $P(Z_i=\bar{Z})>0$ and $P(Z_i=\underbar{Z})>0$, and the denominator is non-zero. This can be shown via a Corollary to Theorem \ref{thmid} presented in Appendix \ref{withcovariates}. By (\ref{waldeq}), a very simple consistent estimator of the ACLATE is thus: $\widehat{ACLATE} := \frac{\hat{\mathbbm{E}}[Y_i|Z_i=\bar{Z}]-\hat{\mathbbm{E}}[Y_i|Z_i=\underline{Z}]}{\hat{\mathbbm{E}}[D_i|Z_i=\bar{Z}]-\hat{\mathbbm{E}}[D_i|Z_i=\underline{Z}]}$. It turns out that $\widehat{ACLATE}$ is in fact numerically equivalent in finite sample to $\hat{\Delta}_c=\hat{\rho}((1\dots 1)')$ obtained via Eq. (\ref{rhonoreg}).\footnote{To see this, note that the vector $H$ of $H_i$ solves the system of equations $\Gamma' H_i = (1 \dots 1)'$. Among vectors that are in the column space of $\Gamma$, $H$ is the unique such solution, given that the design matrix $\Gamma$ has full column rank. One can readily verify that $\Gamma' H = (1\dots 1)$ with the choice $H_i = \frac{\mathbbm{1}(Z_i=(1\dots1))}{\hat{P}(Z_i=(0\dots0))} -\frac{\mathbbm{1}(Z_i=(0\dots0))}{\hat{P}(Z_i=(0\dots0))}$, and that this $H = \Gamma \eta$ with $\eta=(1/\hat{P}(Z_i=(1\dots1)), 0, \dots 0, -1/\hat{P}(Z_i=(0\dots0)))$.} In situations where $Z_i$ has non-zero but small probability for the points $\bar{Z}$ and $\underbar{Z}$, we may thus expect that $\hat{\Delta}_c$ may perform poorly as an estimator of ACLATE in small samples, since it effectively ignores all of the data for which $Z_i \notin \{\underline{Z},\bar{Z}\}$. This issue also arises in the context of IAM \citep{Frolich2007}, in which case $\hat{\rho}_{\bar{Z}, \underline{Z}}$ is also consistent for the ACLATE with finite $\mathcal{Z}$.\footnote{\label{iamfn}An analogous result to Eq. (\ref{waldeq}) holds under IAM with finite instruments, where in that case we take any $\bar{Z} \in \textrm{argmax}_{z}\mathbbm{E}[D_i|Z_i=\mathbf{z}]$ and $\underbar{Z} \in \textrm{argmin}_{z}\mathbbm{E}[D_i|Z_i=\mathbf{z}]$, and define $\mathcal{G}^c:= \{g \in \mathcal{G}: \mathbbm{E}[D_g(Z_i)] \in (0,1)\}$.} Regularization of the estimator to make use of other points in the support of $Z_i$---at the expense of some finite-sample bias---may be useful in improving performance in such cases.\\

\noindent \textit{Covariates:} In Appendix \ref{withcovariates}, I describe how covariates can be accommodated in estimation when instrument independence holds only after conditioning on observed variables $X$. The main result is that while conditional average treatment effects $\Delta_c(x):=\mathbbm{E}[Y_i(1)-Y_i(0)|C_i=c,X_i=x]$ can be identified for each $x$ in the support of $X_i$, the unconditional $\Delta_c$ can be easier to estimate. A particularly simple case occurs when the conditional expectation functions $\mathbbm{E}[Y_i|Z_i=\mathbf{z},X_i=x]$ and $\mathbbm{E}[D_i|Z_i=\mathbf{z},X_i=x]$ are each additively separable between $\mathbf{z}$ and $x$, and linear in $x$. A simple consistent estimator of $\Delta_c$ is then: $ \left((0,\hat{\mathbf{\lambda}}')(\Gamma'\mathcal{M}_X\Gamma)^{-1}\Gamma'\mathcal{M}_XD\right)^{-1}(0,\hat{\mathbf{\lambda}}')(\Gamma'\mathcal{M}_X\Gamma)^{-1}\Gamma'\mathcal{M}_XY$, where $\mathcal{M}_X$ is an orthogonal projection matrix for observations of $X_i$. In this case, the only modification to $\hat{\Delta}_c$ required is to add $X_i$ as additional regressors to the linear projections of $Y_i$ and $D_i$ onto the instruments $\Gamma_i$. I implement this estimator in the empirical application below.

\section{Empirical application} \label{empiricalapp}

In this section I apply the results of this paper to a well-known setting in which multiple instruments have been used: the labor market returns to college. While most existing literature has in this setting bases IV methods on the traditional IAM notion of monotonicity (or on an assumption of homogeneous treatment effects), I instead base estimates on the identification results of this paper that hold under VM. This approach reveals new evidence of heterogeneity in treatment effects across groups that differ in their counterfactual selection behavior, under more plausible assumptions. 

I use the dataset from \citet*{Carneiro2011} (henceforth CHV) constructed from the 1979 National Longitudinal Survey of Youth. This setting is also considered by MTW (under assumption of PM). The sample consists of 1,747 white males in the U.S., first interviewed in 1979 at ages that ranged from 14 to 22, and then again annually. The outcome of interest $Y_i$ is the log of individual $i$'s wage in 1991, and treatment $D_i=1$ indicates $i$ attended at least some college. As in CHV, treatment effects are expressed in approximate per-year equivalents by dividing them by four.

CHV consider four separate instruments for schooling. In a baseline setup, I use the two binary instruments discussed throughout this paper: tuition and proximity. In particular, I let $Z_{1i}$ indicate that average tuition rates local to $i$'s residence around age 17 falls below the sample median, which corresponds to about \$2,170 in 1993 dollars. I let $Z_{2i}$ indicate the presence of a public college in $i$'s county of residence at age 14. I later add two additional instruments used by CHV, which capture local labor market conditions when a student is in high school.

While VM is a natural assumption for the tuition and proximity instruments, a conditional version of instrument validity is more plausible than Assumption 1. I follow CHV and include a set of control variables $X_i$,\footnote{In particular, a student's corrected Armed Forces Qualification Test score, mother's years of education, number of siblings, ``permanent'' local earnings in county of residence at 17, ``permanent'' unemployment in county of residence at 17, earnings in county of residence in 1991, and unemployment in state of residence in 1991, along with an indicator for urban residence at 17 and cohort dummies (see CHV for variable definitions and construction). Also following CHV, I include as components of $X_i$ the squares of continuous control variables. All together, these represent the union of variables that CHV use in their first stage and outcome equation, with one exception: As MTW do, I drop years of experience in 1991 since it may itself be affected by schooling. In the two instrument setup, I also add to $X_i$ the two ``unused'' instruments from CHV and their squares: long-run local earnings in county of residence at 17 and long run unemployment in state of residence at 17.} implemented as described in Section \ref{est} and Appendix \ref{withcovariates}. Standard errors are computed by applying the delta method to the system of estimated regression equations (allowing for heteroscedasticity and cross-correlation between the equations).
 
\subsection{Results from baseline setup with two instruments}
The left panel of Table \ref{baselineFS} reports a cross tabulation of the two instruments, which have a weak positive correlation, though the observations are fairly evenly distributed across the four cells. 

\begin{table}[htp!]
	\quad \quad \quad \quad
	\begin{tabular}{cccc}
		\multicolumn{4}{c}{\underline{Distribution of the instruments}} \\ \\
		& & \multicolumn{2}{c}{$Z_2$=``close''} \\
		\multirow{3}{*}{\vspace{-.5cm}$Z_1$=``cheap''} &  & \textbf{0} & \textbf{1} \\ \cline{3-4} 
		& \multicolumn{1}{l|}{\textbf{0}} & \multicolumn{1}{c|}{469} & \multicolumn{1}{c|}{401} \\ \cline{3-4} 
		& \multicolumn{1}{l|}{\textbf{1}} & \multicolumn{1}{c|}{361} & \multicolumn{1}{c|}{516} \\ \cline{3-4} 
	\end{tabular}
	\vspace{.5cm}
	\quad \quad \quad \quad
	\begin{tabular}{cccc}
		\multicolumn{4}{c}{\underline{Mean fitted propensity scores}} \\
		\\
		& & \multicolumn{2}{c}{$Z_2$} \\
		\multirow{3}{*}{\vspace{-.5cm}$Z_1$} &  & \textbf{far} & \textbf{close} \\ \cline{3-4} 
		& \multicolumn{1}{l|}{\textbf{expensive}} &  \multicolumn{1}{c|}{0.451} & \multicolumn{1}{c|}{0.509} \\ \cline{3-4} 
		& \multicolumn{1}{l|}{\textbf{cheap}} & \multicolumn{1}{c|}{0.487} & \multicolumn{1}{c|}{0.530} \\ \cline{3-4} 
	\end{tabular}
	\vspace{.5cm}
	\caption{Left: number of observations having each pair of values of the instruments, with total sample size $N=1,747$. Right: fitted propensity scores estimated by OLS, evaluated at the sample mean of the $X_i$ variables.} \label{baselineFS}
\end{table}

The right panel of Table \ref{baselineFS} reports predictions from the estimated conditional propensity score function $\mathbbm{E}[D_i|Z_i=\mathbf{z}, X_i=x]$ estimated via a linear regression of $D_i$ on the instruments (and their interaction) as well as $X_i$, then evaluated at the mean $\bar{x}$ of $X_i$.\footnote{I note that when all controls are omitted from this regression, the estimated propensity score function is no longer monotonic in $Z_1$ and $Z_2$. This underscores the potential of VM to be used to evaluate the validity of instruments given a set of conditioning variables (in contrast to PM and IAM, who lack this testable implication).}

The top-left value of $\hat{\mathcal{P}}(expensive, far,\bar{x}) = 45.1\%$ provides an estimate of the overall proportion of always-takers in the population, while the share of never-takers is estimated to be $1-0.53=47.0\%$. The remaining roughly $8\%$ of the population are generalized ``compliers'' consisting of the tuition ($Z_1$), proximity ($Z_2$), eager and reluctant compliers. From the table we can also see that $P(D_i(expensive,close,x)>D_i(expensive,far,x))$ is estimated to be $5.7\%$, and $P(D_i(cheap,far,x)>D_i(expensive,far,x))$ is estimated to be $3.6\%$ (these quantities are the same for all values of $x$, under the maintained assumption that $\mathbbm{E}[D_i|Z_i,X_i]$ is additively separable between $Z_i$ and $X_i$). Combining these figures and the response group definitions from Section \ref{comono}, we see that between 1.5\% and 3.6\% of the population are estimated to be eager compliers, while no more than 2.1\% are reluctant compliers. Similarly, no more than 3.6\% are tuition compliers, and between 2.1\% and 5.7\% are proximity compliers. 

Figure \ref{estimates_2iv} reports estimates of several of the parameters introduced in Section \ref{idsec}, alongside fully-saturated 2SLS for comparison. Consider first the ACLATE: the point estimate of $0.14$ indicates that having attended a year of college increases 1991 wages of all compliers by roughly 14\% on average. This estimate is within the range of roughly $-0.1$ to $0.3$ of the marginal treatment effect (MTE) function estimated by CHV under the assumption of IAM, and is similar to their point estimate of the average treatment on the treated under a parametric normal selection model. The 2SLS estimate from Figure \ref{estimates_2iv} yields a similar value at $0.12$. Note that given the limited sample size none of the estimates are quite significant at the 90\% level. I focus discussion on the point estimates for the sake of illustration with this caveat.
\begin{figure}[htp!]
	\centering
	\includegraphics[width=4in]{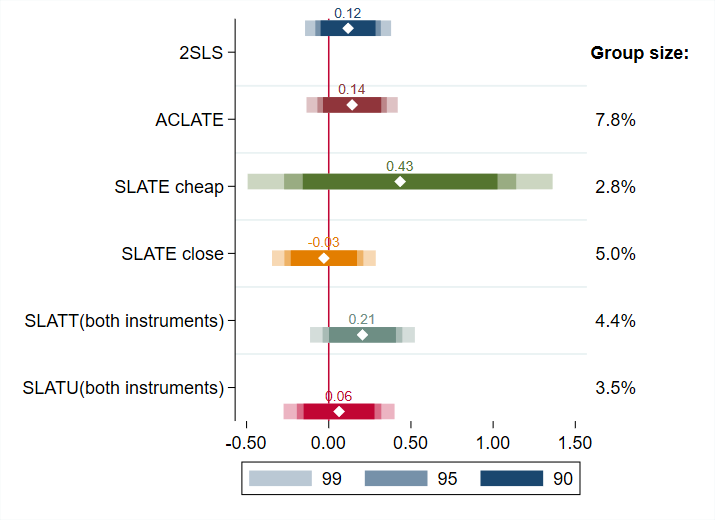} \vspace{.1cm}
	\caption{Estimates of various causal parameters identified under VM with two instruments, alongside fully-saturated 2SLS for comparison. Bars indicate 95\% confidence intervals, and ``Group Size'' refers to the identified quantity $P(C_i=1)$ for each parameter.} \label{estimates_2iv}
\end{figure}

The point estimates from the remaining rows in Figure \ref{estimates_2iv} suggest that the ACLATE aggregates over substantial heterogeneity in the population. For example, the proximity SLATE indicates that a year of college has no average effect on the wages of individuals who move into treatment if and only if a college is nearby, given local affordability. This group includes proximity compliers, eager compliers for whom college is expensive, and reluctant compliers for whom it is cheap. On the other hand, the SLATE for tuition is about three times as large as the ACLATE. These results suggest that the average treatment effect among tuition compliers is larger than it is among proximity compliers, however the sign of the difference is not identified.\footnote{Note however that in the $J=2$ case, if $\Delta_g$ and the corresponding group size $p_g$ is known ex-ante for one group $g \in \mathcal{G}^c$, then the remaining three group specific treatment effects and group sizes can be point identified.} Note finally that the point estimates for $SLATU$ and $SLATT$ suggest that among the compliers averaged over by the ACLATE, those who in fact go to college have greater treatment effects on average than those who do not, which is consistent with students selecting on the basis of their heterogeneous future gains (as in a Roy-type model).

\subsection{Results with all four instruments}
I now add the additional two instruments from CHV, to increase comparability and emphasize the scalability of my method to several instruments. Let $Z_{3i}$ indicate that local earnings in $i$'s county of residence at 17 is below the sample median, and $Z_{4i}$ that unemployment in $i$'s state of residence at 17 is above the sample 25\% percentile (this threshold is chosen as it yields a stronger predictor of college as compared with using the median). The two labor market variables and their squares are removed from the controls $X_i$. 

\begin{figure}[htp!]
	\centering
	\includegraphics[width=4in]{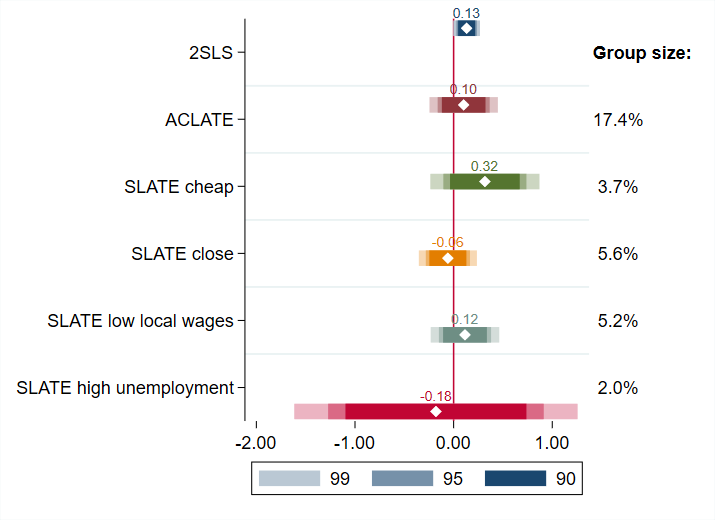} \vspace{.1cm}
	\caption{Estimates of various causal parameters identified under VM with all four instruments, alongside fully-saturated 2SLS for comparison. Bars indicate 95\% confidence intervals, and ``Group Size'' refers to the identified quantity $P(C_i=1)$ for each parameter.} \label{estimates_4iv}
\end{figure}
With all four instruments, over 17\% of the population are now some type of ``complier'' and counted in $\mathcal{G}^c$, which now contains 167 underlying response groups (compared with just 7.8\% of the population for the four groups in $\mathcal{G}^c$ with the two instruments used before). Nevertheless, computing the treatment effect estimates involves regressions with at most 16 terms in addition to the controls, keeping implementation manageable.

Table \ref{estimates_4iv} shows that the $ACLATE$ is not much changed from the case with only two instruments, and we again have that the tuition SLATE is much larger and that the proximity SLATE is close to zero. The SLATE for low local wages occupies an intermediate value, while the SLATE for high unemployment is estimated to be negative (but with a much larger standard error). The unemployment SLATE is so imprecisely estimated in part because its corresponding complier group is the smallest of the estimands considered: with just $2\%$ of the population.

\section{Conclusion}

In this paper, I have characterized the causal parameters that can be point identified using multiple instruments under a monotonicity assumption that is often motivated by economic theory: vector monotonicity (VM). I accomplish this by focusing on binary instruments, but results are applicable to discrete instruments more generally, as shown in Appendix \ref{alternative}.

The estimator I propose targets well-defined causal-parameters at no additional computational cost relative to the popular 2SLS estimator, which is not guaranteed to recover an interpretable causal parameter under VM. In an application to the labor market returns to college education, I find leveraging VM that underlying groups in the population which exhibit different selection behavior also have very different average returns to college.

%% file: vm_appendices.tex
\section{Identification and estimation with covariates} \label{withcovariates}
In practice, it is often easier to justify a \textit{conditional} version of Assumption 1: 
$$\left\{(Y_i(1), Y_i(0), G_i) \indep Z_i\right\} | X_i$$
in which $X$ is a vector of observed covariates unaffected by treatment. VM is assumed as before, i.e. $D_i(\mathbf{z}) \ge D_i(\mathbf{z}')$ for all $i$, whenever $\mathbf{z} \ge\mathbf{z}'$ componentwise. This implies that we may continue to take the ``1'' value of each (binary) instrument to be the direction in which potential treatments are increasing, regardless of the value of $X_i$. This section discusses how one can accommodate these covariates when estimating causal effects.

If full instrument support (Assumption 3) holds conditional $X_i=x$ for each $x$, then Theorem \ref{thmid} implies that we can identify $\Delta_c(x):=\mathbbm{E}[\Delta_i|C_i=1, X_i=x]$ for $\Delta_c$ satisfying Property M, from the distribution of $(Y_i,Z_i,D_i)|X_i=x$. The function $h$ from Theorem \ref{thmid} must now depend as well on the conditioning value $x$: $h(Z_i,x) =\mathbf{\lambda}(x)' Var(\Gamma_i|X_i=x)^{-1}\left(\Gamma_i - \mathbbm{E}[\Gamma_i|X_i=x]\right)$, where we define $\mathbf{\lambda}_g(x) := P(C_i=1|G_i=g,X_i=x)=\mathbbm{E}[c(g,Z_i)|X_i=x]$ for each simple response group $g \in \mathcal{G}^s$. Note that $\lambda_g(x)$ is identified (just as $\lambda_g$ was in the unconditional case) from the data given the known function $c$. Theorem \ref{thmid} applied to the conditional distribution of observables $(Y_i,D_i,Z_i)$ given $X_i=x$ implies that $\Delta_c(x)= \mathbbm{E}[h(Z_i,x)Y_i|X_i=x]/\mathbbm{E}[h(Z_i,x)D_i|X_i=x]$. 

If the support of $X_i$ reflects a small number of discrete values, it might be feasible to repeat the entire estimation on fixed-covariate subsamples, to estimate $\Delta_c(x)$ for each $x$. If $X_i$ includes continuous variables, estimation of $\Delta_c(x)$ can in principle be implemented by nonparametric regression of each component of $\Gamma_i$ on $X_i$ as well as nonparametrically estimating the conditional variance matrix $Var(\Gamma_i|X_i=x)$ (\citet{Yin2010} describe a kernel-based method for this). The vector $\mathbf{\lambda}(x)$ could also be computed via nonparametric regression of $c(g,Z_i)$ on $X_i$ for each $g \in \mathcal{G}^s$.

However, when the object of interest is the \textit{unconditional} version of $\Delta_c$, the conditional quantities become nuisance parameters. Notably, they can be integrated over separately in the numerator and the denominator of the above expression for $\Delta_c(x)$. To see this, note that:
\begin{align*}
	\Delta_c &= \int dF_{X|C=1}(x)\cdot\frac{\mathbbm{E}[h(Z_i,x)Y_i|X_i=x]}{\mathbbm{E}[h(Z_i,x)D_i|X_i=x]}=  \int dF_{X|C=1}(x)\cdot \frac{\mathbbm{E}[h(Z_i,x)Y_i|X_i=x]}{P(C_i=1|X_i=x)}\\
	&= \frac{1}{P(C_i=1)}\int dF_{X}(x)\cdot \mathbbm{E}[h(Z_i,X_i)Y_i|X_i=x] = \frac{\mathbbm{E}[h(Z_i,X_i)Y_i]}{\mathbbm{E}[h(Z_i,X_i)D_i]}
\end{align*}
applying Bayes' rule and using that $P(C_i=1|X_i=x) = \mathbbm{E}[h(Z_i,x)D_i|X_i=x]$. The above expression provides a vector monotonicity analog to a similar result that holds under IAM \citep{Frolich2007}, and suggests an alternative to integrating over conditional estimates $\hat{\Delta}_c(x)$ which may be attractive when $X$ has continuous components or otherwise takes a large number of values.

In the unconditional case, the following corollary to Theorem \ref{thmid} expresses $\Delta_c$ in terms of conditional expectation functions of $Y_i$ and $D_i$ on the instruments (see the proof of Theorem \ref{thmid}):
\begin{corollary*}[to Theorem \ref{thmid}] \label{cefcorr}
	Under the Assumptions of Theorem \ref{thmid}:
	$$\Delta_c = \frac{\sum_{\mathbf{z} \in \mathcal{Z}} \left(\sum_{g \in \mathcal{G}^s}\lambda_g \mathcal{A}_{g\mathbf{z}}\right)\mathbbm{E}[Y_i|Z_i=\mathbf{z}]}{\sum_{\mathbf{z} \in \mathcal{Z}}\left(\sum_{g \in \mathcal{G}^s}\lambda_g \mathcal{A}_{g\mathbf{z}}\right)\mathbbm{E}[D_i|Z_i=\mathbf{z}]}$$
	where $\lambda_g$ is as defined in Theorem \ref{thmid} and $\mathcal{A}_{g\mathbf{z}} = \mathbbm{1}(\mathbf{z}_1 \subseteq S(g))\cdot (-1)^{|S(g)-\mathbf{z}_1|}$, with $(\mathbf{z}_1,\mathbf{z}_0)$ a partition of the indices $j \in \{1 \dots J\}$ that take a value of zero or one in $\mathbf{z}$, respectively.	
\end{corollary*}
\noindent In a setting with covariates, the expectations condition on $X_i$ as well and we have instead that:
\begin{align*}
	\Delta_c &= \frac{\mathbbm{E}[\tilde{\mathbf{\lambda}}(X_i)'\mathcal{A}\left\{\mathbbm{E}[Y_i|Z_i= \mathbf{z}, X_i]\right\}]}{\mathbbm{E}[\tilde{\mathbf{\lambda}}(X_i)'\mathcal{A}\left\{\mathbbm{E}[D_i|Z_i= \mathbf{z}, X_i]\right\}]}
\end{align*}
where $\mathcal{A}$ is a $2^J \times 2^J$ matrix defined from the entries $\mathcal{A}_{g\mathbf{z}}$ above, $\tilde{\mathbf{\lambda}}(x)$ is a vector over $\{a.t.\} \cup \mathcal{G}^s$ with components $\lambda_g(x)$ for $g \in \mathcal{G}^s$ and $0$ for the always-takers, and $\left\{\cdot\right\}$ indicates a vector over $\mathbf{z} \in \mathcal{Z}$. If the CEFs of $Y$ and $D$ are both be separable between $Z$ and $X$, i.e $\mathbbm{E}[Y_i|Z_i= \mathbf{z}, X_i=x] = y(\mathbf{z})+w(x)$ and $\mathbbm{E}[D_i|Z_i= \mathbf{z}, X_i=x] = d(\mathbf{z})+v(x)$, then this expression simplifies to:
\begin{align*}
	\Delta_c &= \frac{\mathbbm{E}[\tilde{\mathbf{\lambda}}(X_i)'\mathcal{A}\left\{y(\mathbf{z})\right\}+w(X_i)\tilde{\mathbf{\lambda}}(X_i)'\mathcal{A} \mathbf{1}]}{\mathbbm{E}[\tilde{\mathbf{\lambda}}(X_i)'\mathcal{A}\left\{d(\mathbf{z})\right\}+v(X_i)\tilde{\mathbf{\lambda}}(X_i)'\mathcal{A} \mathbf{1}]} =\frac{\mathbbm{E}[\tilde{\mathbf{\lambda}}(X_i)]'\mathcal{A}\left\{y(\mathbf{z})\right\}}{\mathbbm{E}[\tilde{\mathbf{\lambda}}(X_i)]'\mathcal{A}\left\{d(\mathbf{z})\right\}}
\end{align*}
where $\mathbf{1}$ is a vector of ones over $\mathbf{z} \in \mathcal{Z}$ and we have used that $\tilde{\mathbf{\lambda}}(x)'\mathcal{A}\mathbf{1}=0$ for any $x$.\footnote{This follows from the definition of the entries: $\mathcal{A}_{g\mathbf{z}} = \mathbbm{1}(\mathbf{z}_1 \subseteq S(g))\cdot (-1)^{|S(g)-\mathbf{z}_1|}$ where $\mathbf{z}_1$ is the set of components of $\mathbf{z}$ that are equal to one. The identity $\sum_{S' \subseteq S} (-1)^{|S'|} = \mathbbm{1}(S = \emptyset)$ for any set $S$ implies that for any $g \in \mathcal{G}^s$ with $S(g)=S$, $[\mathcal{A}\mathbf{1}]_g=\sum_{\mathbf{z}:\mathbf{z}_1 \subseteq S} (-1)^{|S-\mathbf{z}_1|} = (-1)^{|S|} \sum_{S' \subseteq S} (-1)^{|S'|} = 0$. The first component of $\mathcal{A}\mathbf{1}$, corresponding to $g=a.t.$, does not contribute to $\tilde{\mathbf{\lambda}}(x)'\mathcal{A}\mathbf{1}$ for any $x$ since the first component of $\tilde{\mathbf{\lambda}}(x)$ is defined to be zero. The same identity used above annihilates all but two components of $\tilde{\mathbf{\lambda}'}\mathcal{A}$ in Corollary \ref{cefcorr} in the case of the ACLATE, yielding Eq. (\ref{waldeq}).} Note that $\mathbbm{E}[c(g,Z_i)]=\mathbbm{E}[\mathbbm{E}[c(g,Z_i)|X_i]]=\mathbbm{E}[\lambda_g(X_i)]$. Thus we can write $\Delta_c = \frac{\tilde{\mathbf{\lambda}}'\mathcal{A}\left\{y(\mathbf{z})\right\}}{\tilde{\mathbf{\lambda}}'\mathcal{A}\left\{d(\mathbf{z})\right\}}$, where $\tilde{\lambda}$ is constructed using the unconditional quantities $\lambda_g:=\mathbbm{E}[c(g,Z_i)]$. 

Given consistent estimators $\hat{y}(\mathbf{z})$ and $\hat{d}(\mathbf{z})$ of the functions $y(\mathbf{z})$ and $d(\mathbf{z})$, we can estimate $\Delta_c$ by $\hat{\Delta}_c = \frac{\sum_{\mathbf{z} \in \mathcal{Z}} \left(\sum_{g \in \mathcal{G}^s}\hat{\mathbf{\lambda}}_g \mathcal{A}_{g\mathbf{z}}\right)\hat{y}(\mathbf{z})}{\sum_{\mathbf{z} \in \mathcal{Z}}\left(\sum_{g \in \mathcal{G}^s}\hat{\mathbf{\lambda}}_g \mathcal{A}_{g\mathbf{z}}\right)\hat{d}(\mathbf{z})}$ using the unconditional estimators $\hat{\mathbf{\lambda}}_g$ from Table \ref{tablelatesest}. For example, one can use OLS regressions on $Z_i$ and $X_i$ when the functions $y, d, w$ and $v$ are all linear in their arguments. The estimates reported in the empirical application of Section \ref{empiricalapp} use this result, with $w(x)$ and $v(x)$ taken to each be linear in $x$. Note that since the vector $\Gamma_i$ contains a full set of interactions between the binary instruments, both $y(\mathbf{z})$ and $d(\mathbf{z})$ are automatically linear in $\Gamma_i$. When the functions $w(x)$ and $v(x)$ are also linear in $x$, $\hat{\Delta}_c = \left((0,\hat{\mathbf{\lambda}}')(\Gamma'\mathcal{M}_X\Gamma)^{-1}\Gamma'\mathcal{M}_XD\right)^{-1}(0,\hat{\mathbf{\lambda}}')(\Gamma'\mathcal{M}_X\Gamma)^{-1}\Gamma'\mathcal{M}_XY$ yields a consistent estimator of $\Delta_c(x)$, where $\mathcal{M}_X$ is the orthogonal projection matrix for the design matrix of $X_i$, composed of $n$ observations of $X_i$ arranged as rows. Comparing with $\hat{\rho}(\hat{\mathbf{\lambda}})$ from Section \ref{est}, the inclusion of $\mathcal{M}_X$ simply residualizes the $\Gamma_i$ with respect to their linear projection on $X_i$. In practice, the only change required to accommodate covariates in this case is to augment the linear projections of $Y_i$ and $D_i$ onto the instruments with $X_i$ as additional linear regressors.

\section{Identification and estimation without rectangular support} \label{alternative}
	This section provides an extension of Theorem 1 for cases when the support $\mathcal{Z}$ of the instruments is not rectangular (i.e. $\mathbb{S}_Z \ne \mathcal{Z}_1 \times \mathcal{Z}_2 \times \dots \times \mathcal{Z}_J$), and there may be perfect linear dependencies between the instruments.
	
	One instance in which this extension might be applied is if one begins with instruments that are not binary. The following Proposition shows that if one starts with finite discrete instruments satisfying vector monotonicity, these discrete instruments can be re-expressed as a larger number of binary instruments in a way that preserves VM (while preserving all information about the value of $Z_i$): 
	\begin{proposition} \label{propdtob}
		Let $Z_1$ be discrete with $M+1$ ordered points of support $z_0 < z_1 < \dots < z_M$, and $Z_2 \dots Z_J$ be other instruments. Define $\tilde{Z}_{mi} := \mathbbm{1}(Z_{1i} \ge z_m)$. If the vector $Z=(Z_1, \dots Z_J)$ satisfies Assumption VM on a non-disjoint $\mathcal{Z}$ then so does the vector $(\tilde{Z}_1, \dots ,\tilde{Z}_{M}, Z_2, \dots Z_J)$.
	\end{proposition}
	\begin{proof}
	See Online Appendix.
	\end{proof}
	 \noindent Applying Proposition \ref{propdtob} iteratively allows one to begin with discrete instruments in a given empirical setting, and then replace them with a set of binary instruments that still satisfy VM. This is done by introducing one binary instrument per value for any instrument that was initially discrete, omitting the lowest value for each initial instrument.

	However, the construction of Proposition \ref{propdtob} does imply that the support $\mathbb{S}_Z$ of the instruments will not be rectangular, violating an assumption of Theorem \ref{thmid}. Let $\tilde{\mathcal{Z}}$ be the set of values the instruments can take after the transformation of Proposition \ref{propdtob} is applied to e.g. $Z_1 \in \{0,1,2\}$ and $Z_2\in\{0,1\}$. Then with the new set of instruments $(\tilde{Z}_2,\tilde{Z}_3,Z_2)$, we cannot have e.g. $(0,1,0) \in \tilde{\mathcal{Z}}$ or $(0,1,1) \in \tilde{\mathcal{Z}}$ because this would require both $\mathbbm{1}(Z_{1i} \ge 2)=1$ and $\mathbbm{1}(Z_{1i} \ge 1)=0$.\\
	
	\noindent A weaker version of Assumption 3 that allows for such non-rectangular support among binary instruments consists of the following two conditions. Define $Z_{Si}=\prod_{j\in S} Z_{ji}$ where we let $Z_{\emptyset i}:=1$.
	\begin{assumption*} [3a* (existence of instruments)] 
		There exists a family $\mathcal{F}$ of subsets of the instruments $S \subseteq \{1 \dots J\}$, where $\emptyset \in \mathcal{F}$ and $|\mathcal{F}|>1$, such that random variables $Z_{Si}$ for all $S \in \mathcal{F}$ are linearly independent, i.e. $P\left(\sum_{S \in \mathcal{F}} \omega_S \cdot Z_{Si}=0\right)<1$ for all vectors $\mathbf{\omega} \in \mathbbm{R}^{|\mathcal{F}|}/\mathbf{0}$, where $\mathbf{0}$ denotes the zero vector in $\mathbbm{R}^{|\mathcal{F}|}$.
	\end{assumption*}

	\begin{assumption*} [3b* (non-redundant sets of instruments generate the response groups)]
	There exists a family $\mathcal{F}$ satisfying Assumption 3a*, such that for any $S \notin \mathcal{F}$, $g(F) \notin \mathcal{G}$ for all Sperner families $F$ that contain $S$. 
	\end{assumption*}
	\noindent Assumption 3a* alone is very weak, and is satisfied whenever there exists \textit{some} product of the instruments that has strictly positive variance. Assumption 3b* is much more restrictive: it implies that selection functions $D_g(\cdot)$ for all response groups $g \in \mathcal{G}^c$ can be generated from those linearly independent simple selection groups $D_g(\cdot)$ for which $S(g) \in \mathcal{F}$. Assumption 3 corresponds to the special case in which $\mathcal{F}$ is the family of \textit{all} $2^J$ subsets of $J$ binary instruments.
		
	The following Proposition shows that the construction in Proposition \ref{propdtob} mapping discrete instruments to binary instruments yields a case where both parts of Assumption 3* hold, if the original discrete instruments have rectangular support:	
	\begin{proposition} \label{ass3star}
		Let each original instrument $Z_j$ have $M_j+1$ ordered points of support $\mathcal{Z}_j = \{z^j_0,z^j_1,\dots z^j_{M_j}\}$, where $z^j_0 < z^j_1 \dots < z^j_{M_j}$. Define $\tilde{Z}^j_{m} = \mathbbm{1}(Z_{ji} \ge z^j_m)$ and $\mathcal{T}:=\{(j,m)\}_{\substack{j \in \{1\dots J\}\\ m = 1 \dots M_j}}$. If the support of the original discrete instruments is rectangular, i.e. $\mathbb{S}_Z = (\mathcal{Z}_1 \times \mathcal{Z}_2 \times \dots \times \mathcal{Z}_J)$, then Assumption 3* holds with $\mathcal{F}$ the family of all subsets $S$ of $\mathcal{T}$ built as follows: for each $j \in \{1 \dots J\}$, $S$ contains either $(j,m)$ for no values $m$ or all $(j,m)$ for $m$ between $1$ and $m_j$.
	\end{proposition}
	\begin{proof}
		See Online Appendix.
	\end{proof}
	\noindent To interpret the notation of Proposition \ref{ass3star}, note that our ``effective'' binary instruments $\tilde{Z}_{m}^j$ obtained after applying Proposition \ref{propdtob} are indexed by pairs $(j,m)$. One can read $(j,m) \in S$ as saying that the set $S$ ``contains'' the binary instrument $\tilde{Z}_{m}^j$. Consider e.g. a case in which a discrete instrument $Z_1$ has three levels $\{0,1,2\}$ and instruments $2$ to $J$ are each already binary. Proposition \ref{propdtob} shows that if $Z_1 \dots Z_J$ satisfies VM then so does the set of $J+1$ instruments $\tilde{Z}_1, \tilde{Z}_2, Z_2, \dots Z_J$ where $\tilde{Z}_1 = \mathbbm{1}(Z_{1i} \ge 1)$ and $\tilde{Z}_2 = \mathbbm{1}(Z_{1i} \ge 2)$. In this example, the family $\mathcal{F}$ from Proposition \ref{ass3star} would correspond to all subsets of $\{\tilde{Z}_1, \tilde{Z}_2, Z_2, \dots Z_J\}$ that do not contain $\tilde{Z}_2$ without also containing $\tilde{Z}_1$.\footnote{If instead we used the full powerset $\mathcal{F} = 2^{\{1 \dots J\}}$ there would be $2^{J-1}$ ``redundant'' simple response groups in the vector $\Gamma_i = \{\Gamma_{Si}\}_{S \in \mathcal{F}}$, since for any $S \subseteq \{2\dots J\}$: $\tilde{Z}_{2i} \tilde{Z}_{3i} Z_{Si} = \tilde{Z}_{3i}Z_{Si}$.} Intuitively, an element of $\mathcal{F}$ amounts to choosing for each instrument exactly one of its values $m_j$.\footnote{For example, if $J=3$, the subset $\{\tilde{Z}_2,Z_3\}$ (i.e. $S=\{(1,1),(1,2),(3,1)\}$ in the notation of Proposition \ref{ass3star}) would correspond to $Z_1=2$, $Z_2=0$, and $Z_3=1$. There exists an isomorphism between $\mathcal{F}$ and all combinations $(m_1, m_2 \dots m_J)$ of values of the original instruments, as explained in the proof of Proposition \ref{ass3star}.} Provided rectangular support on the original instruments, Assumption 3* then follows for $\mathcal{F}$ constructed in this way, by Proposition \ref{ass3star}.
	
    Given Assumption 3*, Theorem \ref{thmid} generalizes as follows:
	\begin{theorem*}[1*] \label{thmidstar}
		The results of Theorem \ref{thmid} hold under Assumption 3* replacing Assumption 3, where now $\Gamma_i := \{Z_{Si}\}_{S \in \mathcal{F}, S \ne \emptyset}$, $\mathbf{\lambda}:= \{\mathbbm{E}[c(g(S),Z_i)]\}_{S \in \mathcal{F}, S \ne \emptyset}$.
	\end{theorem*}
	\begin{proof}
		Theorem 1* is established in the main proof of Theorem \ref{thmid}, but some steps require more involved calculations under Assumption 3*, with details given in the Online Appendix.
	\end{proof}
	\noindent Theorem 1* may be useful when discrete instruments are mapped to binary instruments as in Proposition \ref{propdtob}, but also in other cases in which the practitioner has auxiliary knowledge that some of the response groups are not present in the population, or are ruled out on conceptual grounds.\footnote{Note that a parameter $\Delta_c$ that satisfies Property M when $\mathcal{Z}$ is rectangular (such as the ACLATE) may violate Property M when $\mathcal{Z}$ is not (to verify Property M in a given empirical context, Proposition \ref{propmsuff} may be useful). Further, a causal parameter is only well-defined under Property 3* if $c(g(F),\mathbf{z})=0$ for all $\mathbf{z} \in \mathcal{Z}$, for any $F$ that contains $S \notin \mathcal{F}$. That is, the function $c$ cannot place weight on groups that are assumed not to exist.}  When it comes to estimation, the matrix $\Gamma$ from Section \ref{est} can be defined from $Z_{Si}$ using only the sets $S$ within $\mathcal{F}$ (c.f. footnote \ref{fn:equivalently}), and similarly for $\hat{\lambda}$ as a vector with $|\mathcal{F}|-1$ components.

%% file: vm_mtwappendix.tex
\section{Comparison with the identification approach of MTW2} \label{sec:compare}

This section compares the point identification results of this paper to the approach to identification proposed by \citet{Mogstad2020b} (MTW2). While MTW2's method is applicable more generally under PM, I focus here on the application of their method when the additional restriction of VM holds. For comparison with my Theorems \ref{thmid} and \ref{necctheorem}, I also assume that the instruments are binary with full support (with the VM order $\geq_j$ for each instrument taken to be the standard order on the real numbers). For simplicity, I focus on treatment effect parameters of the form $\Delta_c$ in this section, rather than individual counterfactual means $\mu_c^d$.

The main result of this section is that in such a setting treatment effect parameters of the form $\Delta_c = \mathbbm{E}[Y_i(1)-Y_i(0)|c(G_i,Z_i)=1]$ are point identified by my Theorem \ref{thmid} if and only if they are point identified by the approach of MTW2, when the approach of MTW2 is employed with a ``full'' set of identifying moments (and no additional identifying assumptions). My results can therefore be interpreted as providing characterizing exactly \textit{which} treatment effect parameters are point-identified under the MTW2 approach,\footnote{However, I note that some parameters that are identified under VM---for example the ACLATE---are arguably more natural to define using the framework of the present paper, which defines target parameters in terms of the full selection groups $G_i$, rather than from single-instrument marginal response (MTR) functions as MTW2 do.} My Theorem \ref{thmid} also yields a constructive estimand for identified parameters that affords simple estimation and statistical inference, while also giving the researcher knowledge the parameter is point-identified, ex-ante (before seeing the data).

The results of this section also offer a partial answer to a question left as an open one by MTW2: whether the identified sets delivered by their method are sharp. With binary instruments satisfying VM, I find that the approach of MTW2 does deliver sharp identified sets when the target parameter satisfies Property M. However, their approach can return empty identified sets when it is used to impose additional assumptions regarding marginal treatment response (MTR) functions that turn out to be incompatible with the data.

\subsection{Identification in MTW2}

The approach to identification used by MTW2 builds upon the idea of \textit{IV-like estimands} of \citet*{Mogstad2018}. For any known measurable function $s(d,\mathbf{z})$, the quantity $\beta_s = \mathbbm{E}[s(D_i,Zi)Y_i]$ is identified from the data and is referred to as an \textit{IV-like estimand}. Let $\mathcal{S}$ denote a collection of IV-like estimands $\{\beta_s\}_{s \in \mathcal{S}}$. 

Given a set $\mathcal{S}$ of IV-like estimands to be used for identification, the identified set proposed by MTW2 for a parameter of interest $\beta^*(m)$ is $$B^{MTW}(\mathcal{S}):=\{\beta^*(m): m \in (\mathcal{M} \cap \mathcal{M}^{obs(\mathcal{S})} \cap \mathcal{M}^{lc(\mathcal{S})})\}$$
where $\mathcal{M}$, $\mathcal{M}^{obs(\mathcal{S})}$ and $\mathcal{M}^{lc(\mathcal{S})}$ are each sets of $m$, where $m$ denotes a collection of MTR functions.\footnote{An MTR function is $\mathbbm{E}[Y_i(d)|U_{ji}=u,Z_{-j,i}=\mathbf{z}_{-j}]$ viewed as a function of $u$, and $m$ collects these functions across $d,j$ and $\mathbf{z}_{-j}$. The Online Appendix reviews how the variables $U_{ji}$ are defined by MTW2.} In particular, i) $\mathcal{M}$ is the set of $m$ that comport with any maintained assumptions about the MTR functions (e.g. that they are monotonic, or satisfy other shape constraints such as concavity); ii) $\mathcal{M}^{obs(\mathcal{S})}$ is the set of $m$ that recover the correct values of $\beta_s$ for all $s \in \mathcal{S}$; and iii) $\mathcal{M}^{lc(\mathcal{S})}$ is the set of $m$ that satisfy a condition called ``mutual consistency'' for each $s \in \mathcal{S}$. $B^{MTW}(\mathcal{S})$ is the set of all values that $\beta(m)$ can take among the $m$ that satisfy all three of these conditions. A review of how MTR functions are defined in MTW2 is provided in the Online Appendix.

Define the functions $s_{d,\mathbf{z}}(d',\mathbf{z}') = \mathbbm{1}(d=d',\mathbf{z}=\mathbf{z}')$ and 
and let $\bar{\mathcal{S}}:=\left\{s_{d,\mathbf{z}}\right\}_{d \in \{0,1\},\mathbf{z} \in \mathcal{Z}}$ be the corresponding set of IV-like estimands for these functions. I refer to $\bar{\mathcal{S}}$ as the ``canonical set'' of IV-like estimands. \citet{Mogstad2018} study $\bar{\mathcal{S}}$ under IAM. Proposition \ref{expectid} in Appendix \ref{proofs} shows that a target parameter is point identified from \textit{some} finite set of IV-like estimands and the observable joint distribution $\mathcal{P}_{DZ}$ of $D_i$ and $Z_i$ if and only if it is identified from $\bar{\mathcal{S}}$ and $\mathcal{P}_{DZ}$. $\bar{\mathcal{S}}$ also provides a basis for all IV-like estimands in the sense that for any other measurable function $s(d,\mathbf{z})$: $\beta_s = \sum_{\mathbf{z} \in \mathcal{Z}} \sum_{d =0}^1 P(D_i=d,Z_i=\mathbf{z}) \cdot \beta_{s_{d,\mathbf{z}}}$ where $\beta_{s_{d,\mathbf{z}}} = \mathbbm{E}[s_{d,\mathbf{z}}(D_i,Z_i)Y_i]$.

\subsection{Equivalence under VM for point-identified $\Delta_c$}
The following result concerns the case in which the parameter of interest $\beta^*$ takes the form $\Delta_c$ considered in this paper (defined in terms of the response groups $G_i$ and instruments $Z_i$). It makes use of the notation developed in the proof of Theorem \ref{thmid}. In particular, let $x_{dg} := P(G_i=g)\cdot \mathbbm{E}[Y_i(d)|G_i =g]$, $[A(d)]_{\mathbf{z}g}=\mathbbm{1}(D_g(\mathbf{z})=d)$, and $b_{d\mathbf{z}} = \mathbbm{E}[Y_i\mathbbm{1}(D_i=d)|Z_i=\mathbf{z}]$. As in Theorem \ref{thmid} define matrix $A$ to have components 
$[A(d)]_{\mathbf{z}g}$ for all $d \in \{0,1\}, \mathbf{z} \in \mathbb{S}_Z$, and $g \in \mathcal{G}$, $\mathbf{x}$ be a vector with components $x_{dg}$ across $d \in \{0,1\}, g \in \mathcal{G}$, and $\mathbf{b}$ a vector with components $b_{d\mathbf{z}}$ for all $d \in \{0,1\}, \mathbf{z} \in \mathbb{S}_Z$. Since the matrix $A$ depends on $\mathcal{G}$ and $\mathbb{S}_Z$, the set of response functions, let us for clarity denote as $A^{VM}$ the $2^{J+1} \times 2\cdot Ded_J$ matrix $A$ that applies under VM with $J$ binary instruments and full rectangular support. Although the method of MTW2 assumes only PM and not VM, I will characterize the set $B^{MTW}(\bar{\mathcal{S}})$ in terms of $A^{VM}$, when VM in fact holds.

Finally, note that given Assumption 1 a parameter of the form $\Delta_c$ can be written as $\Delta_c = \mathbf{\theta_c}'\mathbf{x}$, where $\mathbf{\theta_c}$ is a conformable vector with value $(-1)^{d+1}\frac{\mathbbm{E}[c(g,Z_i)]}{\mathbbm{E}[c(G_i,Z_i)]}$ for component $d,g$.

\begin{theorem} \label{thm:equiv}
	Consider a target parameter of the form $\beta^* = \Delta_c$. Under Assumptions 1-3:
	$$B^{MTW}(\bar{\mathcal{S}}) \subseteq \{\mathbf{\theta_c}'\mathbf{x}: A^{VM}\mathbf{x} = \mathbf{b}\}$$
\end{theorem}
\noindent The proof of Theorem \ref{thm:equiv} is given in the Online Appendix. It follows from Theorem \ref{thm:equiv} that if $\{\mathbf{\theta_c}'\mathbf{x}: A^{VM}\mathbf{x} = \mathbf{b}\}$ returns a singleton then $B^{MTW}(\bar{\mathcal{S})}$ is either a singleton or the empty set (corresponding to model misspecification through the restrictions embedded in $\mathcal{M}$). The proof of Theorem \ref{thmid} shows that the set $\{\mathbf{\theta_c}'\mathbf{x}: A^{VM}\mathbf{x} = \mathbf{b}\}$ is a singleton when Property M holds (which given full instrument support implies that $\theta_\mathbf{c}$ lies in the row-space of the matrix $A^{VM}$). Thus any parameter $\Delta_c$ satisfying Property M (and hence identified by Theorem \ref{thmid}) is a parameter for which $B^{MTW}(\bar{\mathcal{S}})$ is either a singleton or empty.

Theorem \ref{necctheorem} of this paper shows that if $\Delta_c$ is point identified from IV-like estimands when one makes only the Theorem \ref{thmid} assumptions, then $\Delta_c$ must satisfy Property M. Thus, combining Theorems 1-3, we see that $B^{MTW}(\bar{\mathcal{S}})$ returns a singleton for a parameter of the form $\Delta_c$ if and only if $\Delta_c$ satisfies Property M (if one implements MTW2's approach without any additional restrictions embedded through $\mathcal{M}$, and making the Theorem 1 assumptions of instrument validity, VM, full support and $P(C_i=1)>0$). If on the other hand $\mathcal{M}$ imposes additional restrictions that are incompatible with the true DGP, then $B^{MTW}(\bar{\mathcal{S}})$ could be empty even if the target parameter satisfies Property M.

%% file: vm_proofs.tex
\section{Proofs} \label{proofs}

See the Online Appendix for proofs of Propositions \ref{VMtest}-\ref{propdtob}.

\subsection{Proof of Lemma \ref{lemmaequiv}}
Lemma \ref{lemmaequiv} is a special case of Proposition \ref{ass3star} of Appendix \ref{alternative} when all of the $J$ instruments are binary and we let $\mathcal{F} = 2^{\{1\dots J\}}$ be the full powerset of $\{1,\dots,J\}$. The proof of Proposition \ref{ass3star} is given in the Online Appendix. 

\subsection{Proof of Theorem \ref{thmid}}
Note that any measurable function $f(Y)$ preserves Assumption 1: i.e. $(f(Y_i(1)), f(Y_i(0)), G_i)$ are jointly independent of $Z_i$, and Assumptions 2-3 are unaffected by such a transformation to the outcome variable. Thus, we can prove Theorem \ref{thmid} with $f(y)=y$ without loss of generality.

This proof is structured in a way that builds a tight connection to the approach to identification in MTW2, and provides intermediate results that support the extended comparison in Appendix \ref{sec:compare} to that paper. A more direct proof of Theorem \ref{thmid} following the intuition described in Section \ref{sec:intuition} can be found in the Online Appendix. The proof below also combines Theorem \ref{thmid} with its generalization Theorem 1* from Appendix \ref{alternative} (which relaxes Assumption 3), clarifying places where the distinction between Assumption 3 and the weaker Assumption 3* is important.

Begin by observing that moments of the form $b(d)_\mathbf{z} = \mathbbm{E}[Y_i\mathbbm{1}(D_i=d)|Z_i= \mathbf{z}]$ imply a system of linear equations that must be satisfied by latent quantities of the form $x(d)_g = P(G_i=g)\mathbbm{E}[Y_i(d)|G_i=g]$. In particular, by the law of iterated expectations and Assumption 1:
$$ b(d)_\mathbf{z} = \sum_{g \in \mathcal{G}} P(G_i=g)\cdot\mathbbm{E}[Y_i\mathbbm{1}(D_g(\mathbf{z})=d)|G_i=g] = \sum_{g \in \mathcal{G}} \mathbbm{1}(D_g(\mathbf{z})=d)\cdot x(d)_g$$
for each $d \in \{0,1\}$ and all $\mathbf{z} \in \mathbb{S}_Z$, where $\mathbb{S}_Z:=\{\mathbf{z} \in \mathcal{Z}: P(Z_i=\mathbf{z})>0\}$ is the support of $Z_i$ (recall that under Assumption 3: $\mathbb{S}_Z=\mathcal{Z} = \{0,1\}^J$). Let $\mathbf{x}(d)$ be a $|\mathcal{G}| \times 1$ vector with elements $x(d)_g$ for all $g \in \mathcal{G}$, and $\mathbf{b}(d)$ a $|\mathbb{S}_Z| \times 1 $ vector with elements $b(d)_\mathbf{z}$ for all $\mathbf{z} \in \mathbb{S}_Z$. Let $A(d)$ be a $|\mathbb{S}_Z| \times |\mathcal{G}|$ matrix of entries of the form $[A(d)]_{\mathbf{z},g}=\mathbbm{1}(D_g(\mathbf{z})=d)$. We now have a system of linear restrictions $A(d)\mathbf{x}(d)=\mathbf{b}(d)$ for each $d \in \{0,1\}$, which we can combine to write $A\mathbf{x}=\mathbf{b}$, where $\mathbf{x}:=(\mathbf{x}(0)',\mathbf{x}(1)')'$, $\mathbf{b}:=(\mathbf{b}(0)',\mathbf{b}(1)')'$, and $A$ is a block diagonal $2|\mathbb{S}_Z| \times 2|\mathcal{G}|$ matrix composed of $A(0)$ and $A(1)$.\footnote{In the discussion preceding Theorem \ref{thm:equiv}, the entries of $\mathbf{x}$ and $\mathbf{b}$ are denoted in the alternative notation $x_{dg}$ and $b_{d\mathbf{z}}$, for brevity there. Theorem \ref{thm:equiv} also uses the notation $A^{VM}$ for $A$ in the special case that the instruments have full support.}

Since the matrix $A$ has more rows than columns under VM (because $|\mathcal{G}|>|\mathcal{Z}| \ge |\mathbb{S}_Z|$), we cannot hope to invert the system $A\mathbf{x}=\mathbf{b}$ to solve for a unique value of $\mathbf{x}$. Point identification instead relies on the parameter of interest taking on the same value for all $\mathbf{x}$ that are compatible with the system of linear equations. A standard result characterizing the solutions to linear systems (see e.g. \citealt{ben2003generalized}) says that the set of vectors $\mathbf{x}$ compatible with $A\mathbf{x}=\mathbf{b}$ can be written as $\{A^+ \mathbf{b} + (I-A^+ A)\mathbf{w}\}$ across all vectors $\mathbf{w} \in \mathbbm{R}^{2|\mathcal{G}|}$, where $A^+$ is the Moore-Penrose pseudo-inverse of $A$.

Conveniently, we can also write a causal parameter of the form $\mu^{d}_c$ as a linear function of $\mathbf{x}$ (generalizing Equation \ref{deltaw}). By the law of iterated expectations over $g$ and Assumption 1, $\mu^{d}_c = \mathbf{\theta}'\mathbf{x}(d)$, where $\mathbf{\theta}$ is a $|\mathcal{G}| \times 1$ vector with elements $\theta_g = \frac{\mathbbm{E}[c(g,Z_i)]}{\mathbbm{E}[c(G_i,Z_i)]}$. To establish a common notation that also covers conditional average treatment effects, let $\Delta_{\alpha,c}:=\alpha_1\cdot \mathbbm{E}[Y_i(1)|C_i=1] + \alpha_0\cdot \mathbbm{E}[Y_i(0)|C_i=1]$ for any $\alpha_0$ and $\alpha_1$ in $\mathbbm{R}$. We can write any such $\Delta_{\alpha,c}$ as $\mathbf{\theta_\alpha}'\mathbf{x}$, where $\mathbf{\theta_\alpha} = (\alpha_0 \cdot \mathbf{\theta}',\alpha_1 \cdot \mathbf{\theta}')'$. This notation allows us to simultaneously nest as special cases i) treatment effects $\Delta_c$ when $\alpha_0 = -1, \alpha_1 = 1$; ii) untreated counterfactual means $\mu^{0}_c$ when $\alpha_0=1, \alpha_1=0$; and similarly iii) treated counterfactual means $\mu^{1}_c$ when $\alpha_0=0, \alpha_1=1$. 

Given the above, the set of values for $\Delta_{\alpha,c}$ that are compatible with the system $A\mathbf{x}=\mathbf{b}$ is
\begin{equation*} \label{eq:BCalpha}
	B_{c,\alpha}:=\{\mathbf{\theta_\alpha}'A^+ \mathbf{b} + \mathbf{\theta_\alpha}'(I-A^+ A)\mathbf{w}\}
\end{equation*}
across $\mathbf{w} \in \mathbbm{R}^{2|\mathcal{G}|}$. The set $B_{c,\alpha}$ is a singleton when $\mathbf{\theta_\alpha}'(I-A^+ A)\mathbf{w}=0$ for all $\mathbf{w}$, which occurs if and only if $\theta_\alpha'$ belongs to the row-space of the matrix $A$ (in this case $\mathbf{\theta_\alpha}'A^+ A=\mathbf{\theta_\alpha}'$). $\Delta_{\alpha,c}$ is then identified provided that $\mathbf{\theta_\alpha}'A^+\mathbf{b}$ is.

Let $D$ be a $|\mathbb{S}_Z| \times |\mathcal{F}|$ matrix with entries $D_{\mathbf{z},S} = \mathbbm{1}(S \subseteq \mathbf{z}_1)$, where $\mathcal{F}$ is a family of subsets of $\{1 \dots J\}$ satisfying Assumption 3* (when the stronger Assumption 3 holds we let $\mathcal{F} = 2^{\{1\dots J\}}$, the full powerset of $\{1 \dots J\}$). Here and in the subsequent proofs we use the notation of Footnote \ref{fn:equivalently}, that $\Gamma_i= (Z_{S_1i} \dots, Z_{S_ki})'$ for some arbitrary ordering of the $k:=|\mathcal{F}|-1$ non-empty subsets $S \in \mathcal{F}$, where $Z_{Si} := \prod_{j \in S} Z_{ji}$. Similarly denote the components of $\lambda$ as $\lambda_S$ for $S \in \mathcal{F}, S \neq \emptyset$ (rather than the equivalent notation $\lambda_g$ for $g$ across $\mathcal{G}^s$ used in the main text).

Let $\mathbf{d}$ be a $|\mathbb{S}_Z|$-vector with elements $d_\mathbf{z} = \mathcal{P}(\mathbf{z})$, and $\mathbf{\tilde{\lambda}} := (0,\mathbf{\lambda}')'$. Theorem \ref{thmid} follows from the following Proposition: 
\begin{proposition} \label{prop:rowspace} 
	If Assumptions 2 and 3* hold, then $\theta_\alpha'$ belongs to the row-space of the matrix $A$ for any $\alpha_0,\alpha_1$. As a result, $\Delta_{c,\alpha}=\mathbf{\theta_\alpha}'A^+ \mathbf{b}=\frac{\mathbf{\tilde{\lambda}}'D^+ \left\{\alpha_1\mathbf{b}(1)-\alpha_0\mathbf{b}(0)\right\}}{\mathbf{\tilde{\lambda}}'D^+ \mathbf{d}}$ and $P(C_i=1)=\mathbf{\tilde{\lambda}}'D^+\mathbf{d}$ provided that $P(C_i=1)>0$ and Assumption 1 holds. 
\end{proposition}

\noindent Letting $\alpha_d=0$ and $\alpha_{1-d}=0$ for either $d \in \{0,1\}$, the Proposition yields identification of $\mu^{d}_c$ as $(-1)^{d+1}\frac{\mathbf{\tilde{\lambda}}'D^+ \mathbf{b}(d)}{\mathbf{\tilde{\lambda}}'D^+ \mathbf{d}}$.

To obtain the form $\Delta_{c,\alpha}=\frac{\mathbf{\tilde{\lambda}}'D^+ \left\{\alpha_1\mathbf{b}(1)-\alpha_0\mathbf{b}(0)\right\}}{\mathbf{\tilde{\lambda}}'D^+ \mathbf{d}}$ written in Theorem \ref{thmid}, let $\mathfrak{Z}_i$ be a $|\mathbb{S}_Z| \times 1$ vector of indicators $\mathbbm{1}(Z_i= \mathbf{z})$ for each of the values $\mathbf{z} \in \mathbb{S}_Z$. With probability one:
$[(1,\Gamma_i')']_S = D_{Z_i,S} = \sum_{\mathbf{z} \in \mathbb{S}_Z}  [\mathfrak{Z}_i]_\mathbf{z} \cdot D_{\mathbf{z},S} =[D' \mathfrak{Z}_i]_S$ for any $S \in \mathcal{F}$. Therefore, $\Sigma^*:=\mathbbm{E}[(1,\Gamma_i')'(1,\Gamma_i')] = D'\mathbbm{E}[\mathfrak{Z}_i\mathfrak{Z}_i']D = D' P D$, where $P$ is a diagonal $|\mathbb{S}_Z| \times |\mathbb{S}_Z|$ matrix with entries $P_{\mathbf{z},\mathbf{z}}=P(Z_i= \mathbf{z})$ for each $\mathbf{z} \in \mathbb{S}_Z$. In Appendix \ref{alternative}, the function $h(z)$ from Theorem \ref{thmid} is generalized under Assumption 3* replacing Assumption 3 to take the same form $h(z) = \mathbf{\lambda}'\Sigma^{-1}(\Gamma_i - \mathbbm{E}[\Gamma_i])$, but with the vector $\Gamma$ (from Section \ref{est}) now defined using only the non-empty sets $S$ within $\mathcal{F}$ rather than from the full powerset $2^{\{1 \dots J\}}$. Thus under either Assumption 3 or Assumption 3* $\Sigma:=Var(\Gamma_i)$ is $k \times k$, where $k=|\mathcal{F}|-1$. We can now write $h(Z_i)$ from Theorem \ref{thmid} as 
\begin{align*}
	h(Z_i) &= \mathbf{\lambda}'\Sigma^{-1}(\Gamma_i - \mathbbm{E}[\Gamma_i]) = \mathbf{\tilde{\lambda}}'{\Sigma^*}^{-1}(1,\Gamma_i')' = \mathbf{\tilde{\lambda}}'(D'PD)^{-1} D' \mathfrak{Z}_i =\mathbf{\tilde{\lambda}}'(D'PD)^{-1} D'P P^{-1} \mathfrak{Z}_i 
\end{align*}
where the second equality can be shown by applying the $2 \times 2$ block inversion formula.

\begin{lemma} \label{lemma:dplus}
	Under Assumption 3*, $(D'PD)^{-1}D'P = D^+$
\end{lemma}
\begin{proof}
	When Assumption 3 holds, this follows immediately from the fact that the matrix $D$ is then invertible (as shown in the proof of Proposition \ref{ass3star}). Thus $(D'PD)^{-1}D'P = D^{-1}P^{-1}\cancel{{D'}^{-1}D}'P = D^{-1}$, and $D^+=D^{-1}$ when $D^{-1}$ exists. A proof that the equality $(D'PD)^{-1}D'P = D^+$ holds more generally under Assumption 3* can be found in the Online Appendix. 
\end{proof}

\noindent Using Lemma \ref{lemma:dplus}, we then have that $h(Z_i)=\mathbf{\tilde{\lambda}}'D^+ P^{-1} \mathfrak{Z}_i$, and thus for any random variable $V_i$:
\begin{align*}
	\mathbbm{E}[h(Z_i)V_i] &= \mathbf{\tilde{\lambda}}'D^+  P^{-1}\mathbbm{E}[\mathfrak{Z}_iV_i]
	=\sum_{\mathbf{z} \in \mathbb{S}_Z} [\mathbf{\tilde{\lambda}}'D^+]_\mathbf{z}\cdot  P(Z_i=\mathbf{z})^{-1}\cdot \mathbbm{E}[\mathbbm{1}(Z_i=\mathbf{z})\cdot V_i]\\
	&=\sum_{\mathbf{z} \in \mathbb{S}_Z} [\mathbf{\tilde{\lambda}}'D^+]_\mathbf{z}  \cdot \mathbbm{E}[V_i|Z_i=\mathbf{z}]:=\mathbf{\tilde{\lambda}}'D^+ \{\mathbbm{E}[V_i|Z_i= \mathbf{z}]\}_{\mathbf{z} \in \mathbb{S}_Z},
\end{align*} 
i.e. $\mathbf{\tilde{\lambda}}'D^+$ describes the coefficients in an expansion of $\mathbbm{E}[h(Z_i)V_i]$ into CEFs of $V_i$ across the support of $Z_i$. Applying this to the variables $Y_i \mathbbm{1}(D_i=d)$ and $D_i$, we arrive at
$$ (-1)^{d+1} \frac{\mathbbm{E}[Y_ih(Z_i)\mathbbm{1}(D_i=d)]}{\mathbbm{E}[h(Z_i)D_i]} = (-1)^{d+1}\frac{\mathbf{\tilde{\lambda}}'D^+ \{\mathbbm{E}[Y_i\mathbbm{1}(D_i=d)|Z_i=\mathbf{z}]\}}{\mathbf{\tilde{\lambda}}'D^+ \{\mathbbm{E}[D_i|Z_i=\mathbf{z}]\}} (-1)^{d+1}\frac{\mathbf{\tilde{\lambda}}'D^+ \mathbf{b}(d)}{\mathbf{\tilde{\lambda}}D^+ \mathbf{d}} = \mu^{d}_c, $$
using Proposition \ref{prop:rowspace}. This establishes Theorem 1* of Appendix \ref{alternative} with Theorem \ref{thmid} as a special case. This also establishes the Corollary to Theorem \ref{thmid} in Appendix \ref{withcovariates} by observing that $\mathbbm{E}[Y_i|Z_i=\mathbf{z}] = \mathbf{b}(0)_\mathbf{z}+\mathbf{b}(1)_\mathbf{z}$, and thus $\Delta_c = \frac{\mathbf{\tilde{\lambda}}'\mathcal{A} \{\mathbbm{E}[Y_i|Z_i=\mathbf{z}]\}}{\mathbf{\tilde{\lambda}}'\mathcal{A} \{\mathbbm{E}[D_i|Z_i=\mathbf{z}]\}}$, where $\mathcal{A}$ is defined in the proof of Lemma \ref{lemma:columnspace} below (under Assumption 3, $D^+ = \mathcal{A}$, as shown therein). 

\subsubsection{Proof of Proposition \ref{prop:rowspace}}

First, observe that since $A$ is block-diagonal, $A^+$ is a block diagonal $2|\mathcal{G}| \times 2|\mathcal{Z}|$ matrix composed of $A(0)^+$ and $A(1)^+$ (effectively, we have a separate system $A(d)\mathbf{x}(d)=\mathbf{b}(d)$ for each $d \in \{0,1\}$). We can thus write $A^+ \mathbf{b}$ as $ A(0)^+ \mathbf{b}(0)+A(1)^+ \mathbf{b}(1)$. The following are some basic properties of the pseudo-inverse that will be useful in what follows: if a matrix $A$ has full column-rank (linearly independent columns), then $A^+ = (A'A)^{-1}A'$, and if a square $A$ is invertible $A^+ = A^{-1}$. The pseudo-inverse commutes with transposition, that is ${A'}^{+} = {A^+}'$. If $A=BC$ and $B$ has full-column rank while $C$ has full row-rank, then $A^+ = C^+ B^+$.\footnote{\label{fn:colrow} To show this, note that these conditions imply that $B^+ B = CC^+ = I_n$. Therefore, $C^+ B^+$ satisfies the four defining conditions to be $A^+$ (see e.g. \citealt{ben2003generalized}): i) $AA^+ A = B\cancel{C C^+} \cancel{B^+ B}C = A$; ii) $A^+ A A^+ = C^+ \cancel{B^+ B}\cancel{C C^+} B^+ =A^+$; iii) $A^+ A = C^+ \cancel{B^+ B}C = C^+ C$ is symmetric; and iv) $A A^+ = \cancel{C C^+} \cancel{B B^+} = I$ is symmetric.}

Let $\mathcal{F}$ be a family of subsets of the instruments $\{1 \dots J\}$ that satisfies Assumption 3* from Appendix \ref{alternative}. In the baseline setup in which Assumption 3 holds (full rectangular support), $\mathcal{F} = 2^{\{1 \dots J\}}$, the full powerset of $\{1 \dots J\}$. When $F \subset 2^{\{1 \dots J\}}$, e.g. when the binary instruments lack full rectangular support, we can index the columns of the matrix $M$ introduced in Section \ref{simple} by the members of $\mathcal{F}$ aside from the empty set. This holds without loss of generality because under Assumption 3b* the entries along this column of $M$ would all be equal to zero (the sets $S \notin \mathcal{F}$ do not show up in any $F(g)$ for any $g \in \mathcal{G}$).

Note that since $c(\cdot, \cdot)$ satisfies Property M, $\mathbbm{E}[c(g,Z_i)]= \sum_{S \subseteq \mathcal{F}, S\ne\emptyset} M_{g,S}\cdot \mathbbm{E}[c(g(S),Z_i)] = [M \mathbf{\lambda}]_g$, which can be shown by simply averaging Property M over the distribution of $Z_i$. Let us adopt a notational convention that any vector or matrix with rows indexed by the selection groups $g \in \mathcal{G}$, the first and last rows correspond to always- and never-takers, respectively. By assumption $\mathbbm{E}[c(g,Z_i])]=0$ for each of these two groups under Property M, and we can thus write $\mathbf{\theta} = \frac{1}{\mathbbm{E}[c(G_i,Z_i)]} (0,(M\mathbf{\lambda})',0)'$. 

To represent this in a more compact notation, let $\tilde{\mathbf{\lambda}} = (0,\mathbf{\lambda}')'$ be the $k \times 1$ vector $\mathbf{\lambda}$ prepended with a zero, where $k:=|\mathcal{F}|-1$,  so that $\tilde{\mathbf{\lambda}}$ has a component for each $S \in \mathcal{F}$. Then we can write the following two expressions for $\theta$, both of which will be useful later:
\begin{equation} \label{eq:theta1reps}
	\mathbf{\theta} = \frac{1}{\mathbbm{E}[c(G_i,Z_i)]} \tilde{M}(1)\mathbf{\tilde{\lambda}}=\frac{-1}{\mathbbm{E}[c(G_i,Z_i)]} \tilde{M}(0)\mathbf{\tilde{\lambda}}
\end{equation}
where $\tilde{M}(0)$ and $\tilde{M}(1)$ are the $|\mathcal{G}| \times |\mathcal{F}|$ matrices
$$ \tilde{M}(1) := \begin{pmatrix}
	1 & \underbrace{\mathbf{0}}_{1 \times k}\\
	\underbrace{\mathbf{0}}_{|\mathcal{G}^c| \times 1} & M\\
	0&\underbrace{\mathbf{0}}_{1 \times k}
\end{pmatrix} \quad \quad \quad \quad \quad \quad \quad \tilde{M}(0) := \begin{pmatrix}
	0 & \underbrace{\mathbf{0}}_{1 \times k}\\
	\underbrace{\mathbf{1}}_{|\mathcal{G}^c| \times 1} & -M\\
	1&\underbrace{\mathbf{0}}_{1 \times k}
\end{pmatrix} $$
	and the $\mathbf{0}$'s and $\mathbf{1}$'s are conformable matrices of zeros or ones respectively, as depicted above. A property of the matrices $\tilde{M}(0)$ and $\tilde{M}(1)$ that will be useful is that they both have full column rank:
	\begin{lemma} \label{lemma:fullrank}
		For either $d \in \{0,1\}$, the matrix $\tilde{M}(d)$ has full column rank given Assumption 3*, and thus $\tilde{M}(1)^+ \tilde{M}(1)=\tilde{M}(0)^+ \tilde{M}(0)=I_{|\mathcal{F}|}$, where $I_n$ is the identity matrix in $\mathbbm{R}^n$ 
	\end{lemma}
	\begin{proof}
		Note that for either $d \in \{0,1\}$, the first	column of $\tilde{M}(d)$ is linearly independent from the rest because the other columns all have zero as their first and last entry. What remains to be shown is that the $k$ columns of $M$ are linearly independent from one another. It is sufficient to show that a $k \times k$ sub-matrix of $J$ has full rank. If one takes the $k$ rows of $M$ corresponding to simple Sperner families $g \in \mathcal{G}^s$ with $S(g) \in \mathcal{F}$, then the resulting submatrix of $M$ is the identity $I_k$, which has full rank. Thus $\tilde{M}(d)$ has full column rank.
	\end{proof}
	 
	Meanwhile, given VM (Assumption 2) we can also use the matrices $\tilde{M}(d)$ to write $A(d) = D\tilde{M}(d)'$ for either $d \in \{0,1\}$, where $D$ is the $|\mathbb{S}_Z| \times |\mathcal{F}|$ matrix defined previously with entries $D_{\mathbf{z},S} = D_S(\mathbf{z})$. In words: the selection functions $D_g(\cdot)$---which are represented by the columns $g$ of the matrix $A(d)$---can be generated as a linear expansion in the selection functions $D_S(\cdot)$ for simple compliance groups along with with never-takers (represented by the columns of $D$), using coefficients from the matrix $\tilde{M}(d)$.
	
	This relationship between $A(d)$ and $\tilde{M}(d)$ proves useful in Lemma \ref{lemma:columnspace} below, which establishes that $\mathbf{\theta}'A(d)^+ A(d) = \mathbf{\theta}'$ (i.e. that $\theta$ is the row space of $A(d)$). This implies identification of $\Delta_{c,\alpha}$ since then $\mathbf{\theta_\alpha}'A^+A=\mathbf{\theta_\alpha}'$ and the set $B_{c,\alpha}$ is a singleton. Recall that Assumption 3* is given in Appendix \ref{alternative} and that Assumption 3 from the main text represents a special case of it.
	\begin{lemma} \label{lemma:columnspace}		
		Given Assumptions 2 and 3*: $\mathbf{\theta}'A(d)^+ A(d) = \mathbf{\theta}'$, i.e. $\theta$ is in the row space of $A(d)$, for either $d \in \{0,1\}$.
	\end{lemma}	
\begin{proof}
	 Consider first the baseline case in which the stronger Assumption 3 holds, so that $\mathbb{S}_Z = \mathcal{Z} = 2^{\{1 \dots J\}}$ and $\mathcal{F}$ consists of all subsets of $\{1 \dots J\}$. The proof is more involved when Assumption 3 is relaxed to Assumption 3* with $\mathcal{F} \subset 2^{\{1 \dots J\}}$, and this general case is handled in the Online Appendix.

	When Assumption 3 holds and $\mathcal{F} = 2^{\{1 \dots J\}}$, the matrix $D$ is $2^J \times 2^J$. We begin by showing that this square matrix has an inverse. In particular, define a $2^J \times 2^J$ matrix $\mathcal{A}$ with entries $\mathcal{A}_{S,\mathbf{z}} = \mathbbm{1}(\mathbf{z}_1 \subseteq S)\cdot (-1)^{|S-\mathbf{z}_1|}$ for all $S \subseteq \{1\dots J\}$, where $(\mathbf{z}_1,\mathbf{z}_0)$ is a partition of the indices $j \in \{1 \dots J\}$ that take a value of zero or one in $\mathbf{z}$, respectively.\footnote{Note that this is equivalent to the matrix $\mathcal{A}$ defined in the Corollary to Theorem \ref{thmid} in Appendix \ref{withcovariates}, except that here we label the rows of $\mathcal{A}$ with $S \subseteq \{1 \dots J\}, S \neq \emptyset$ rather than $g \in \mathcal{G}^s$.}
	
	For any two $\mathbf{z}, \mathbf{z}' \in \mathcal{Z}$, we can expand the quantity $\mathbbm{1}(\mathbf{z}' = \mathbf{z})$ out as a polynomial in the instrument indicators as $\mathbbm{1}(\mathbf{z}'= \mathbf{z})
	= \prod_{j \in \mathbf{z}_1} z'_j \prod_{j \in \mathbf{z}_0}(1-z'_j) = \sum_{S \subseteq \mathbf{z}_0} (-1)^{|S|}\cdot \mathbf{z}'_{(\mathbf{z}_1 \cup S)}$. Then $I_{\mathbf{z}',\mathbf{z}} = [D\mathcal{A}]_{\mathbf{z}',\mathbf{z}}$ because
	$$\mathbbm{1}(\mathbf{z}=\mathbf{z}')= \sum_{S \subseteq \mathbf{z}_0} (-1)^{|S|}\cdot z'_{(\mathbf{z}_1 \cup S)} = \sum_{S \subseteq \{1\dots J\}} \left(\mathbbm{1}(\mathbf{z}_1 \subseteq S)\cdot (-1)^{|S-\mathbf{z}_1|}\right)\cdot \mathbf{z}'_{S} = \sum_{S \subseteq \{1\dots J\}} \mathcal{A}_{S,\mathbf{z}} D_{\mathbf{z}',S}$$
	Since $D$ and $\mathcal{A}$ are square, this implies that both are invertible and $D^{-1} = {D}^+ = \mathcal{A}$.
		
	Since $D$ is full rank and $\tilde{M}(d)'$ has full row-rank by Lemma \ref{lemma:columnspace}, we can write $A(d)^+ = {(D\tilde{M}(d))}^+ = \tilde{M}(d)'^+ D^+= \tilde{M}(d)'^+ \mathcal{A}$. This then implies that $A(d)^+ A(d)=\tilde{M}(d)'^+\tilde{M}(d)'$, and hence $\theta$ belongs to the row space of $A(d)$ if and only if $\tilde{M}(d)\tilde{M}(d)^+\mathbf{\theta} = \mathbf{\theta}$ (i.e.  $\mathbf{\theta}$ is in the column space of $\tilde{M}(d)$). That this latter property holds follows immediately from the representation of $\mathbf{\theta}$ from Eq. (\ref{eq:theta1reps}): $\mathbf{\theta} = \frac{(-1)^{d+1}}{\mathbbm{E}[c(G_i,Z_i)]} \tilde{M}(d)\tilde{\mathbf{\lambda}}$. 
\end{proof}\vspace{.25cm}

	\noindent Given Lemma \ref{lemma:columnspace}, we can now establish that $\Delta_{c,\alpha} = \theta_\alpha ' A^+ \mathbf{b}$, since then: $$\mathbf{\theta_\alpha}'(I-A^+ A) = (\alpha_0\cdot \left\{\mathbf{\theta}'-\mathbf{\theta}'A(0)^+ A(0) \right\},\alpha_0\cdot \left\{\mathbf{\theta}'-\mathbf{\theta}'A(1)^+ A(1) \right\})=(\mathbf{0}_{|\mathcal{G}|}',\mathbf{0}_{|\mathcal{G}|}')=\mathbf{0}_{2|\mathcal{G}|}'$$ To simplify $A^+ \mathbf{b}$ and show that it is equivalent to the form given in Proposition \ref{prop:rowspace} note that $\mathbf{\theta}'A(1)^+ \mathbf{b}(1)=\mathbf{\theta}'\tilde{M}(1)'^+ D^+ \mathbf{b}(1)$ and  $\mathbf{\theta}'A(0)^+ \mathbf{b}(0)=\mathbf{\theta}'\tilde{M}(0)'^+ D^+ \mathbf{b}(0)$. Using the first representation of $\theta$ in Eq. (\ref{eq:theta1reps}), we have that:
	$$\mathbf{\theta}'A(1)^+ \mathbf{b}(1)=\frac{1}{\mathbbm{E}[c(G_i,Z_i)]}\mathbf{\tilde{\lambda}}'\cancel{\tilde{M}(1)'\tilde{M}(1)'^+} D^+\mathbf{b}(1)=\frac{1}{\mathbbm{E}[c(G_i,Z_i)]}\cdot\mathbf{\tilde{\lambda}}' D^+\mathbf{b}(1)$$
	where $\tilde{M}(1)'\tilde{M}(1)'^+ = I_{|\mathcal{Z}|}$ by Lemma \ref{lemma:fullrank}. Using the second representation of $\theta$ in Eq. (\ref{eq:theta1reps}): $$\mathbf{\theta}'A(0)^+ \mathbf{b}(0)=\frac{-1}{\mathbbm{E}[c(G_i,Z_i)]}\mathbf{\tilde{\lambda}}'\cancel{\tilde{M}(0)'\tilde{M}(0)'^+} D^+\mathbf{b}(0)=\frac{-1}{\mathbbm{E}[c(G_i,Z_i)]}\cdot \mathbf{\tilde{\lambda}}'D^+\mathbf{b}(0)$$
	using that $\tilde{M}(0)'\tilde{M}(0)'^+ = I_{|\mathcal{Z}|}$ by Lemma \ref{lemma:fullrank}.
	
	If $\mathbbm{E}[c(G_i,Z_i)]$ is known, $\Delta_{c,\alpha}$ is then identified as:
	\begin{equation} \label{eq:pointid}
		\Delta_{c,\alpha} = \theta_\alpha ' A^+ \mathbf{b} = \frac{1}{\mathbbm{E}[c(G_i,Z_i)]}\cdot \mathbf{\tilde{\lambda}}'D^+ \left\{\alpha_1\mathbf{b}(1)-\alpha_0\mathbf{b}(0)\right\}
	\end{equation}
	
	\noindent It only remains to be shown that $\mathbbm{E}[c(G_i,Z_i)]$ is also identified and equal to $\mathbf{\tilde{\lambda}}'D^+ \mathbf{d}$. Since our derivation of (\ref{eq:pointid}) has made no assumptions about the joint of $(Y_i(0),Y_i(1))$, we can consider the special case in which $Y_i(d)=d$ so that $Y_i=D_i$ with probability one, and $\Delta_c = 1$. Applying (\ref{eq:pointid}) to this setting, we have $1 = \frac{1}{\mathbbm{E}[c(G_i,Z_i)]}\cdot \mathbf{\tilde{\lambda}}'D^+ \mathbf{d}$, and hence $\mathbbm{E}[c(G_i,Z_i)]=\mathbf{\tilde{\lambda}}'D^+ \mathbf{d}$, where note that the RHS of this equality depends only on the joint distribution of $Z_i$ and $D_i$.
	
\subsection{An Equivalence Result for Identification}
	
	The proofs of Theorem \ref{necctheorem} and the discussion in Appendix \ref{sec:compare} make use of the following equivalence result. This result uses the definition of identification given in Footnote \ref{fn:identified}, which has the following useful property: if one set $\mathcal{S}$ of empirical estimands can be written as a known function of another set of empirical estimands $\mathcal{S}'$, then a parameter of interest being point identified by $\mathcal{S}$ implies that this same parameter is also point identified by $\mathcal{S}'$.
	
	\begin{proposition} \label{expectid}
		Let the support $\mathcal{Z}$ of the instruments be finite and Assumption 1 hold. Fix a function $c(g,\mathbf{z})$. Let $\mathcal{P}_{DZ}$ denote the joint distribution of $D_i$ and $Z_i$. Then the following are equivalent:
		\begin{enumerate}
			\item{$\Delta_c$ is (point) identified by $\mathcal{P}_{DZ}$ and a finite set of IV-like estimands $\beta_s := \mathbbm{E}[s(D_i,Z_i)Y_i]$, where each function $s(d,\mathbf{z})$ is known or identified from $\mathcal{P}_{DZ}$}
			\item{$\Delta_c=\beta_s$ for a single such $s$}
			\item{$\Delta_c$ is identified from $\mathcal{P}_{DZ}$ and the set of CEFs $\{\mathbbm{E}[Y_i|D_i=d,Z_i= \mathbf{z}]\}_{d \in \{0,1\},\mathbf{z} \in \mathcal{Z}}$}
			\item $\Delta_c$ is (point) identified by $\mathcal{P}_{DZ}$ and the set of IV-like estimands corresponding to the functions in $\bar{\mathcal{S}}:=\left\{s_{d',\mathbf{z}'}\right\}_{d' \in \{0,1\},\mathbf{z}' \in \mathcal{Z}}$, where $s_{d',\mathbf{z}'}(d,\mathbf{z}) = \mathbbm{1}(d=d',\mathbf{z}=\mathbf{z}')$ 
		\end{enumerate}
	\end{proposition}
	\begin{proof}		
	Let $\mathcal{S}$ denote a set of measurable functions $s(d,\mathbf{z})$ defining IV-like estimands $\{\beta_s\}_{s \in \mathcal{S}}$. We can show each of the following implications:
	\begin{itemize}
		\item $\mathbf{2 \rightarrow 1}$ Immediate, since \textbf{2} is a special case of \textbf{1} with $\mathcal{S}$ a singleton	
		\item $\mathbf{3 \rightarrow 1}$ Let $\mathcal{S} = \{s_{d,\mathbf{z}}\}_{d \in \{0,1\}, z \in \mathcal{Z}}$, where $s_{d,\mathbf{z}}(d',\mathbf{z}')=\mathbbm{1}(d'=d, \mathbf{z}' =\mathbf{z})$. Then each $\beta_s$ is equal to $P(D_i=d, Z_i= \mathbf{z})\cdot\mathbbm{E}[Y_i|D_i=d,Z_i= \mathbf{z}]$ for some $d,\mathbf{z}$, where $P(D_i=d, Z_i= \mathbf{z})$ is known from $\mathcal{P}_{DZ}$.
		\item $\mathbf{1 \rightarrow 3}$ Any $\beta_s$ can be written: $\beta_s = \sum_{d,\mathbf{z}} P(D_i=d, Z_i= \mathbf{z}) s(d,\mathbf{z}) \mathbbm{E}[Y_i|D_i=d, Z_i= \mathbf{z}]$, and is thus pinned down by the CEFs $\mathbbm{E}[Y_i|D_i=d, Z_i= \mathbf{z}]$, the joint distribution $\mathcal{P}_{DZ}$, and the known function $s$.
		\item $\mathbf{4 \rightarrow 1}$ Immediate, since $\{\beta_{s_{d',\mathbf{z}'}}\}_{d' \in \{0,1\},\mathbf{z}' \in \mathcal{Z}}$ is a finite set of IV-like estimands. 
		\item $\mathbf{3 \rightarrow 4}$ from the proof of $\mathbf{3 \rightarrow 1}$, we saw that each $\beta_{s_{d,\mathbf{z}}}=P(D_i=d, Z_i= \mathbf{z}) \cdot \mathbbm{E}[Y_i|D_i=d,Z_i= \mathbf{z}]$, and the denominator is known from $\mathcal{P}_{DZ}$.
		\item $\mathbf{3 \rightarrow 2}$ Note that given knowledge of $\mathcal{P}_{DZ}$, knowing the set of CEFs $\{\mathbbm{E}[Y_i|D_i=d,Z_i= \mathbf{z}]\}_{d \in \{0,1\},\mathbf{z} \in \mathcal{Z}}$ is equivalent to knowing the vector $\mathbf{b}$ having components $\mathbbm{E}[Y_i\mathbbm{1}(D_i=d)|Z_i=\mathbf{z}]$, following the notation in the proof of Theorem \ref{thmid}. As shown there, the set of values of $\Delta_c$ compatible with the outcome CEFs can then be written as, using Assumption 1: $\{\mathbf{\theta_c}'A^+ \mathbf{b} + \mathbf{\theta_c}'(I-A^+ A)\mathbf{w}\}_{\mathbf{w} \in \mathbbm{R}^{2|\mathcal{G}|}}$ where $\mathbf{\theta_c}'=(-\mathbf{\theta}',\mathbf{\theta}')$. If Assumptions 1-3 hold, then $A$ is the $2^J \times Ded_J$ matrix $A^{VM}$, in the notation of Theorem \ref{thm:equiv}.
		
		This set must be a singleton for $\Delta_c$ to be identified absent additional restrictions, since otherwise an infinite collection of values of $\Delta_c$ would be compatible with the full set of restrictions $A\mathbf{x}=\mathbf{b}$ placed on $\mathbf{x}$ by the outcome CEFs (given that $\mathbf{\theta_c}$ is not the zero vector). For this set to be a singleton for a given $A$, the vector $\mathbf{\theta_c}$ must lie in the row space of the matrix $A$, so that $\mathbf{\theta_c}'(I-A^+ A)$ is equal to the zero vector.
		
		Thus, by \textbf{3.}, we have that $\Delta_c = \mathbf{\theta_c}'A^+ \mathbf{b}$. But this implies \textbf{2.}, if we take $s(d,\mathbf{z}) = \frac{P(D_i=d|Z_i= \mathbf{z})}{P(D_i=d,Z_i= \mathbf{z})}\cdot [\mathbf{\theta_c}'A^+]_{(d,\mathbf{z})}=\frac{1}{P(Z_i= \mathbf{z})}\cdot [\mathbf{\theta_c}'A^+]_{(d,\mathbf{z})}$, where $[\mathbf{\theta_c}'A^+]_{(d,\mathbf{z})} = [(-1)^{d+1}\cdot \mathbf{\theta}'A(d)^+]_{\mathbf{z}} $ is the component of the vector $\mathbf{\theta_c}'A^+$ corresponding to the pair $(d,\mathbf{z})$. Note that $A^+$ is a known matrix given $\mathcal{G}$, and $\mathbf{\theta_c}$ is a known function of the marginal distribution of $Z_i$, up to the factor $\mathbbm{E}[c(G_i,Z_i)]$. It only remains to be shown that $\mathbbm{E}[c(G_i,Z_i)]$ is also identified under assumption of \textbf{1}. As in the proof of Theorem \ref{thmid}, take the case in which $Y_i(d)=d$ with probability one. Using the result above and that $\mathbf{\theta} = \frac{(-1)^{d+1}}{\mathbbm{E}[c(G_i,Z_i)]} \tilde{M}(d)\tilde{\mathbf{\lambda}}$: $\mathbbm{E}[c(G_i,Z_i)] = \mathbbm{E}[\tilde{s}(Z_i)D_i]$, where $\tilde{s}(\mathbf{z}) := \frac{1}{P(Z_i= \mathbf{z})}\cdot [\tilde{\mathbf{\lambda}}'\tilde{M}(d)'A(1)^+]_{\mathbf{z}}$.
	\end{itemize}
	\end{proof}
	
	\subsection{Proof of Theorem \ref{necctheorem}}
	
	Note that if $\mu^{d}_c$ is identified for all $d$ and measurable function $f$, so must $\Delta_c = \mu^{1}_c-\mu^{0}_c$ with $f(y)=y$. By Proposition \ref{expectid} above, it follows that if $\Delta_c$ is identified from a finite set of IV-like estimands and $\mathcal{P}_{DZ}$,  then it can be written as a single one: $\Delta_c = \beta_s$ with $s(d,\mathbf{z})$ an identified functional of ${P_{DZ}}$. Write $Y_i = Y_i(0)+D_i\Delta_i$ where $\Delta_i:=Y_i(1)-Y_i(0)$. Then, using the law of iterated expectations:
	\begin{align*}
		\Delta_c &= \beta_s 
		=\sum_{g} P(G_i=g)\left\{\mathbbm{E}[s(D_g(Z_i),Z_i)Y_i(0)|G_i=g]+\mathbbm{E}[s(D_g(Z_i),Z_i)D_g(Z_i)\Delta_i|G_i=g]\right\}\\
		&=\sum_{g} P(G_i=g)\left(\cancel{\mathbbm{E}[s(D_g(Z_i),Z_i)]}\right)\mathbbm{E}[Y_i(0)|G_i=g]\\
		&\hspace{2in}+\sum_{g} P(G_i=g)\left(\mathbbm{E}[s(1,Z_i)D_g(Z_i)]\right)\mathbbm{E}[\Delta_i|G_i=g]
	\end{align*}
	\noindent where I've used independence and that $s(D_g(Z_i),Z_i)D_g(Z_i)=s(1,Z_i)D_g(Z_i)$ for all $i$ in the third equality. The crossed out term must be equal to zero, because we've assumed $\beta_s=\Delta_c$ holds for \textit{every} joint distribution of response groups and potential outcomes compatible with the maintained model assumptions (that $\mathbbm{E}[s(D_g(Z_i),Z_i)]=0$ can also be verified directly in the case covered by Theorem \ref{thmid}). Consider two such distributions, identical except that the second distribution describes a case in which for all units $Y_i(0) \rightarrow Y_i(0)+\delta$ and $Y_i(1)\rightarrow Y_i(1)+\delta$. Then the $\Delta_i$ and hence $\Delta_c$ are unchanged, but if the crossed out term were not zero, $\beta_s$ would change as $\mathbbm{E}[Y_i(0)|G_i=g] \rightarrow \mathbbm{E}[Y_i(0)|G_i=g] + \delta$. Thus $\Delta_c=\sum_{g} P(G_i=g)\left(\mathbbm{E}[s(1,Z_i)D_g(Z_i)]\right)\Delta_g$. 
	
	Recall that from Equation (\ref{deltaw}) that $\Delta_c$ can also be written as a weighted average of group-specific average treatment effects $\Delta_g=\mathbbm{E}[Y_i(1)-Y_i(0)|G_i=g]$ as: $ \Delta_c = \frac{1}{P(C_i=1)} \sum_g P(G_i=g)\cdot \mathbbm{E}[c(g,Z_i)] \cdot \Delta_g$. Since $\beta_s=\Delta_c$ holds for any vector of $\{\Delta_g\}_{g \in \mathcal{G}^c}$, we can match coefficients within each group to establish that $\mathbbm{E}[c(g,Z_i)] =P(C_i=1)\mathbbm{E}[s(1,Z_i)D_g(Z_i)]$. This set of weights satisfies Property M, since:
	\begin{align*}
		\mathbbm{E}[c(g,Z_i)] &= P(C_i=1)\cdot \mathbbm{E}\left[s(1,Z_i)\sum_{g' \in \mathcal{G}^s}M_{gg'}D_{g'}(Z_i)\right]\\
		& = \sum_{g' \in \mathcal{G}^s}M_{gg'}\left( P(C_i=1)\cdot \mathbbm{E}[s(1,Z_i)D_{g'}(Z_i)]\right)= \sum_{g' \in \mathcal{G}^s}M_{gg'}\cdot \mathbbm{E}[c(g',Z_i)]
	\end{align*}
	for any $g \in \mathcal{G}^c$. If this holds for any distribution of $Z_i$ satisfying Assumption 3, then we must have $c(g,\mathbf{z})=\sum_{g' \in \mathcal{G}^s}M_{gg'}\cdot c(g',\mathbf{z})$ for all $\mathbf{z} \in \mathcal{Z}, g \in \mathcal{G}^c$. To see this, consider a sequence of distributions for $Z_i$ that converges point-wise to a degenerate distribution at any single point $\mathbf{z}$, but satisfies Assumption 3 for each term in the sequence. Applying the dominated convergence theorem to $\mathbbm{E}[c(g,Z_i)]-\sum_{g' \in \mathcal{G}^s}M_{gg'}\cdot \mathbbm{E}[c(g',Z_i)]=0$ along this sequence, we have that $c(g,\mathbf{z}) - \sum_{g' \in \mathcal{G}^s}M_{gg'} \cdot c(g',\mathbf{z})=0$. A similar argument establishes that $c(a.t.,\mathbf{z})=c(n.t.,\mathbf{z})=0$ for all $\mathbf{z} \in \mathcal{Z}$ given that $\mathbbm{E}[c(g,Z_i)] =P(C_i=1)\cdot \mathbbm{E}[s(1,Z_i)D_g(Z_i)]$ and $\mathbbm{E}[s(1,Z_i)]=0$.